\documentclass[journal,twoside,web]{ieeecolor}
\let\labelindent\relax
\usepackage{enumitem}
\usepackage{amsmath,amsfonts}
\usepackage{generic}
\usepackage{mathtools}
\usepackage{graphicx,cite} 
\usepackage{bm}
\usepackage{hyperref}
\usepackage{caption,subcaption}
\usepackage{multicol,amssymb}
\usepackage{balance}

\newtheorem{theorem}{Theorem}

\definecolor{darkgreen}{rgb}{0.0, 0.4, 0.0}

\DeclareMathAlphabet\mathbfcal{OMS}{cmsy}{b}{n}

\DeclareMathOperator*{\st}{subject~to}

\DeclareMathOperator*{\argmin}{arg\,min}

\newcommand{\norm}[1]{\left\lVert#1\right\rVert}
\newcommand{\normtwo}[1]{\left\lVert#1\right\rVert_{2}}
\newcommand{\frobenius}[1]{\left\lVert#1\right\rVert_{F}}
\newcommand{\hatbf}[1]{\widehat{\mathbf{#1}}}

\newcommand{\starbf}[1]{\mathbf{#1}^\star}
\newcommand{\starhatbf}[1]{\widehat{\mathbf{#1}}^\star}
\newcommand{\starhatbm}[1]{\widehat{\bm{#1}}^\star}

\newcommand{\hatbm}[1]{\widehat{\bm{#1}}}
\newcommand{\starbm}[1]{\bm{#1}^\star}
\newcommand{\tildebm}[1]{\widetilde{\bm{#1}}}
\newcommand{\cbm}[1]{\bm{#1}^c}

\newtheorem{lemma}{Lemma}
\newtheorem{proposition}{Proposition}

\newtheorem{definition}{Definition} 
\newtheorem{remark}{Remark}

\newtheorem{assumption}{Assumption}
\begin{document}

\title{Near-Optimal Design of Safe Output Feedback Controllers from Noisy Data}

\author{Luca Furieri, Baiwei Guo$^\star$, Andrea Martin$^\star$, and Giancarlo Ferrari-Trecate
    \thanks{Authors are with the \'Ecole Polytechnique Fédérale de Lausanne (EPFL), Institute of Mechanical Engineering, CH-1015 Lausanne, Switzerland. E-mails: {\tt\footnotesize \{luca.furieri, baiwei.guo, andrea.martin, giancarlo.ferraritrecate\}@epfl.ch}.}
    \thanks{$^{\star}$Baiwei Guo and Andrea Martin contributed equally to this work.} 
    \thanks{Research supported by the Swiss National Science Foundation under the NCCR Automation (grant agreement 51NF40\textunderscore 80545).}
}

\maketitle
\allowdisplaybreaks
\begin{abstract}
As we transition towards the deployment of data-driven controllers for black-box cyberphysical systems, complying with hard safety constraints becomes a primary concern.
Two key aspects should be addressed when input-output data are corrupted by noise: how much uncertainty can one tolerate without compromising safety, and to what extent is the control performance affected? By focusing on finite-horizon constrained linear-quadratic problems, we provide an answer to these questions in terms of the model mismatch incurred during a preliminary identification phase. We propose a control design procedure based on a quasiconvex relaxation of the original robust problem and we prove that, if the uncertainty is sufficiently small, the synthesized controller is safe and near-optimal, in the sense that the suboptimality gap increases \emph{linearly} with the model mismatch level. Since the proposed method is independent of the specific identification procedure, our analysis holds in combination with state-of-the-art behavioral estimators beyond standard least-squares. The main theoretical results are validated by numerical experiments. 
\end{abstract}


\section{Introduction}
Many safety-critical engineering systems that play a crucial role in our modern society are becoming too complex to be accurately modeled through white-box state-space models \cite{baggio2021data}. As a consequence, most contemporary control approaches envision unknown black-box systems for which a safe and optimal behavior must be attained by solely relying on a collection of system's output trajectories in response to different inputs. 

Controllers for unknown systems can be designed according to two paradigms. \emph{Model-based} methods follow a two-step procedure: first,  data are exploited to identify the system parameters, and then a suitable controller is computed for the estimated model. On the other hand, \emph{model-free} methods aim at directly learning an optimal control policy, without explicitly reconstructing an internal representation of the dynamical system. For a description of advantages and limitations of both approaches, we refer to \cite{recht2019tour}, among recent surveys. 

Given the intricacy of deriving rigorous suboptimality and sample-complexity bounds, most recent model-based and model-free approaches have focused on basic Linear Quadratic Regulator (LQR) and Linear Quadratic Gaussian (LQG) control problems as suitable benchmarks to establish how machine learning  can be interfaced to the continuous action spaces typical of control~\cite{dean2019sample, fazel2018global,zheng2020sample,simchowitz2020improper,lale2020logarithmic,zhang2020policy,tsiamis2020sample}. For complex control tasks, it is more challenging to perform a thorough probabilistic analysis. Recent advances include \cite{dean2019safely,fattahi2020efficient} for constrained and distributed LQR control with direct state measurements, respectively, and \cite{furieri2020learning} for distributed output-feedback LQG.      

Model-based methods may pose a difficulty when it comes to accurately identifying the state-space model of a large-scale system; this is the case, for instance, for complex networked systems such as the power grid, brain and traffic networks \cite{baggio2021data}. A promising data-driven approach that aims at bypassing a parametric state-space description of the system dynamics, while still being conceptually simple to implement for the users, hinges on the \emph{behavioral framework} \cite{willems1997introduction}. This approach has gained renewed interest with the introduction of Data-EnablEd Predictive Control (DeePC) \cite{coulson2019data,coulson2021distributionally,dorfler2021bridging}, which established that constrained output reference tracking can be effectively tackled in a Model-Predictive-Control (MPC) fashion by plugging adequately generated data into a convex optimization problem. The work \cite{de2019formulas} introduces data-driven formulations for some controller design tasks, and \cite{berberich2020data} derives stability guarantees for closed-loop control.

In many scenarios, however, exact data are not available. For instance, data can be corrupted by measurement noise or even by malicious attacks intended at fatally compromising the safety \cite{russo2021poisoning}, the quality, and the reliability of the synthesized control policies. It is therefore essential that data-driven controllers are endowed with robustness guarantees. While some approaches have been suggested in the behavioral framework, e.g. \cite{coulson2021distributionally,alpago2020extended,yin2021maximum,vanwaarde2020noisy}, it remains fairly unexplored how much noise-corrupted data affect the performance and the safety of data-driven control systems.  Recently, \cite{xue2021data,de2021low} have derived suboptimality \cite{de2021low} and sample-complexity \cite{xue2021data} bounds for LQR through direct behavioral formulations based on 1) Linear Matrix Inequalities (LMI) and 2) the System Level Synthesis (SLS) approach, respectively. A limitation is that the internal system states must be measured, which is unrealistic for several large-scale systems \cite{baggio2021data}. Furthermore, while \cite{de2021low} proves that for low-enough noise a high-performing and robustly stabilizing controller can be found, the corresponding suboptimality growth rate is not explicitly derived. To address these open points, \cite{furieri2021behavioral} has formulated a Behavioral Input-Output Parametrization (BIOP) of linear control policies which makes it possible to derive noise-dependent suboptimality analysis for output-feedback LQG, solely based on a non-parametric estimation of the system open-loop responses. The BIOP can exploit any system identification technique such as the least-square methods of \cite{oymak2019non,zheng2020sample} or behavioral  maximum-likelihood (ML) estimation \cite{iannelli2020experiment,yin2021maximum}. In order to deploy these systems in real-world scenarios, it is however important to include safety guarantees in the analysis. All the mentioned works \cite{xue2021data,de2021low,furieri2021behavioral, zheng2020sample} do not include such safety requirements. 

\subsection{Contributions}
We propose a method for designing safe and near-optimal output-feedback control policies for linear systems in finite-horizon. Our approach is solely based on noisy data, and we explicitly characterize the growth rate of the suboptimality as a function of the mismatch between the true and estimated system. First, we develop a new relaxed optimization problem that guarantees safety while robustly accounting for noise-corrupted data. Second, we show that the incurred level of suboptimality converges to zero approximately as a linear function of the model mismatch incurred during a preliminary identification phase. Hence, upon using a consistent system estimator, the proposed controller is near-optimal in the limit of available data growing to infinity. The corresponding analysis differs from that of \cite{furieri2021behavioral,zheng2020sample}, in that a feasible solution to the proposed optimization problem must be characterized analytically while taking the safety constraints into account. In addition to dealing with constraints in an output-feedback setup — which is the main novelty with respect to \cite{zheng2020sample,furieri2021behavioral} and \cite{dean2019safely} — the effect of the \emph{uncertain} initial condition $x_0$ must be explicitly tracked in the cost. Indeed, \cite{zheng2020sample} assumed that $x_0 = 0$ thanks to the considered infinite-horizon setting. On a more general level, our analysis has been inspired by \cite{dean2019safely}, which combined robust control tools with classical identification techniques to ensure safety of unknown systems with suboptimality guarantees when states are fully observed.
As we only have access to noisy output measurements, we exploit an input-output representation of the plant and analyze four different closed-loop responses to understand how process and output measurement noises impact safety and performance.  Suboptimality with respect to the best model-based \emph{open-loop} control input has very recently been analyzed in  \cite{berberich2021linear} as a function of the noise-level. Instead, in the present paper we analyze the suboptimality brought about by closed-loop policies. In particular, we show a linear  growth rate of the suboptimality in terms of the model mismatch level as compared to the ground-truth constrained output-feedback controller.

A preliminary version of this work has recently appeared in the 60\textsuperscript{th} IEEE Conference on Decision and Control \cite{furieri2021behavioral}. Differently from \cite{furieri2021behavioral}, this paper includes safety constraints in the analysis, thus addressing a novel and independent set of challenges and results. Furthermore, this work includes all the technical proofs. Last, new numerical experiments are developed to consider safety constraints and to explicitly include the estimation procedure of \cite{yin2021maximum}. 

\subsection{Paper structure}
Assuming knowledge of the underlying dynamics, Section~\ref{sec:prob_statement} reviews the optimal control problem of interest and its model-based solution. Section~\ref{sec:Noisy_BIOP_constrained} treats the case where we only have access to noisy input and output data; we propose an optimal control problem that ensures safety against bounded model mismatches, and discuss its numerical implementation. Section~\ref{sec:suboptimality} quantifies the suboptimality incurred by our synthesis procedure as a function of the model mismatch. We present numerical experiments in Section~\ref{sec:simulations} and conclude the paper in Section~\ref{sec:conclusions}.

\subsection{Notation}
We use $\mathbb{R}$ and $\mathbb{N}$ to denote the sets of real numbers and non-negative integers, respectively. We use $I_n$ to denote the identity matrix of size $n \times n$ and  $0_{m \times n}$ to denote the zero matrix of size $m \times n$. We write $\mathbf{x} = \operatorname{vec}(x_1,\ldots,x_N) \in \mathbb{R}^{Nn}$ to denote the vector obtained by stacking together the vectors $x_1,\dots,x_N \in \mathbb{R}^{n}$,  and $\mathbf{M}=\text{blkdiag}(M_1,\dots,M_N)$ to denote a block-diagonal matrix with $M_1,\dots,M_N \in \mathbb{R}^{m \times n}$ on its diagonal block entries. For $\mathbf{M} =  \begin{bmatrix}M_1^\mathsf{T}&\dots&M_N^\mathsf{T}\end{bmatrix}^\mathsf{T}$ we define the block-Toeplitz matrix $$\footnotesize \operatorname{Toep}_{m \times n}\left(\mathbf{M}\right) \hspace{-0.075cm} = \hspace{-0.075cm} \begin{bmatrix}
           M_1&0_{m \times n}&\dots&0_{m \times n}\\ M_2&M_1&\dots &0_{m \times n}\\
           \vdots&\vdots&\ddots&\vdots\\
           M_N&M_{N-1}&\dots&M_1
    \end{bmatrix}
    .$$
More concisely, we will write $\operatorname{Toep}(\cdot)$ when the dimensions of the blocks are clear from the context. The Kronecker product between $M \in \mathbb{R}^{m \times n}$ and $P \in \mathbb{R}^{p \times q}$ is denoted as $M \otimes P \in \mathbb{R}^{mp \times nq}$.  For a vector $v \in \mathbb{R}^n$ and a matrix $A \in \mathbb{R}^{m \times n}$ we denote as $\norm{v}_p$, $\norm{A}_p$, their standard $p$-norm and induced $p$-norms, respectively. For a row vector $x \in \mathbb{R}^{1 \times n}$we define $\norm{x}^\star_1 = \sum_{i=1}^n|x_i|$. The Frobenius norm of a matrix $M \in \mathbb{R}^{m \times n}$ is denoted by $\norm{M}_{F}=\sqrt{\text{Trace}(M^\mathsf{T} M)}$. For a symmetric matrix $M$, 
we write $M \succ 0$ or $M \succeq 0$ if it is positive definite or positive semidefinite, respectively. We say that $x\sim \mathcal{D}(\mu,\Sigma)$ if the random variable  $x\in \mathbb{R}^n$ follows a distribution with mean $\mu \in \mathbb{R}^n$ and covariance matrix $\Sigma \in \mathbb{R}^{n \times n}, \Sigma \succeq 0$.

A finite-horizon trajectory of length $T$ is a sequence $\omega(0),\omega(1),\dots, \omega(T-1)$ with $\omega(t) \in \mathbb{R}^n$ for every $t=0,1,\dots, T-1$, which can be compactly written as
\begin{equation*}
    \bm{\omega}_{[0,T-1]} = \begin{bmatrix}\omega^\mathsf{T}(0)&\omega^\mathsf{T}(1)& \dots&\omega^\mathsf{T}(T-1)\end{bmatrix}^\mathsf{T} \in \mathbb{R}^{nT}\,.
\end{equation*}
When the value of $T$ is clear from the context, we will omit the subscript $[0,T-1]$. For a  finite-horizon trajectory $\bm{\omega}_{[0,T-1]}$ we also define the Hankel matrix of depth $L$ as
\begin{equation*}
\footnotesize
    \mathcal{H}_{L}(\bm{\omega}_{[0,T-1]}) = \begin{bmatrix}
           \omega(0)&\omega (1)&\dots&\omega(T-L)\\ \omega(1)&\omega(2)&\dots &\omega(T-L+1)\\
           \vdots&\vdots&\ddots&\vdots\\
           \omega(L-1)&\omega(L)&\dots&\omega(T-1)
    \end{bmatrix}\,.
\end{equation*}

\section{Problem Statement: the Model-Based Case}
\label{sec:prob_statement}
In this section, we review safe output-feedback controller synthesis when the system model is known. We consider a discrete-time linear system with output observations, whose state-space  representation is given by
\begin{equation} \label{eq:dynamic}
    \begin{aligned}
        x(t+1) = A x(t)+B u(t),~~y(t) = C x(t) + v(t)\,,
    \end{aligned}    
\end{equation}
where $x(t) \in \mathbb{R}^n$ is the state of the system and $x(0) =x_0$ for a predefined $x_0 \in \mathbb{R}^n$, $u(t)\in \mathbb{R}^m$ is the control input, $y(t)\in \mathbb{R}^p$ is the observed output, and $v(t)\in \mathbb{R}^p$ denotes  measurement noise $v(t) \sim \mathcal{D}(0,\Sigma_v)$, with $\Sigma_v \succ 0$.  The system is controlled through a time-varying, dynamic affine control policy
\begin{equation}
    \label{eq:input}
    u(t) = \sum_{k=0}^t K_{t,k} y(k)+ g_t+ w(t)\,,
\end{equation}
where $K_{t,k}$ and $g_t$ are the linear and affine parts of the policy, respectively, and  $w(t)\in \mathbb{R}^m$ denotes noise on the input $w(t) \sim \mathcal{D}(0,\Sigma_w)$ with $\Sigma_w \succeq 0$, which acts as process noise.\footnote{The more general model $x(t+1) = A x(t)+B u(t)+w(t)$ would make the cost function depend on a  specific realization $A,B$ explicitly \cite[Chapter 3]{van2012subspace}. Instead, the adopted noise model ensures that the cost only depends on the covariance matrix $\Sigma_w$ and the coordinate-free parameter $\mathbf{G}$, thus making our theoretical bounds meaningful in a data-driven input-output setting.} Furthermore, we assume that the noise is bounded with 
\begin{equation*}
\norm{w}_\infty \leq w_\infty\,, \quad \norm{v}_\infty \leq v_\infty\,,
\end{equation*}
where $w_\infty, v_\infty > 0$. We consider the problem of synthesizing  a feedback control policy that minimizes the expected value with respect to the disturbances of a quadratic objective defined over future input-output trajectories of length $N \in \mathbb{N}$:
    \begin{equation}
\label{eq:cost_output}
J^2:=\mathbb{E}_{w,v}\left[\sum_{t=0}^{N-1}\left(y(t)^\mathsf{T} Q_ty(t)+u(t)^\mathsf{T} R_tu(t)\right)\right]\,,
\end{equation}  
where $Q_t\succeq0$ and $R_t\succ 0$ for every $t = 0,\dots,N-1$.

The problem is made more challenging by the requirement that inputs and outputs satisfy the safety constraints
\begin{equation}
    \label{eq:constraints}
    \begin{bmatrix}y(t)\\u(t)\end{bmatrix} \in \Gamma_t \subseteq \mathbb{R}^{p+m}\,,\quad \forall t = 0, \dots, N-1\,,
\end{equation}
where $\Gamma_t$ is a nonempty polytope for every  $t = 0, \dots, N-1$ defined as
\begin{equation}
    \label{eq:polytope_definition}
    \Gamma_t = \left\{ (y,u) \in (\mathbb{R}^p,\mathbb{R}^m)|~F_y^t y \leq b_y^t, F_u^t u \leq b_u^t\right\}\,,
\end{equation}
with $F_y^t \in \mathbb{R}^{s \times p}$, $F_u^t \in \mathbb{R}^{s \times m}$ and $b_y^t,b_u^t \in \mathbb{R}^s$ for every  $t = 0, \dots, N-1$. Despite \eqref{eq:cost_output} being convex in the input and output trajectories and $\Gamma_t$ being polytopic, we highlight that minimizing \eqref{eq:cost_output} subject to \eqref{eq:dynamic}, \eqref{eq:input} and \eqref{eq:constraints} is a non-convex problem in the control policy parameters $K_{t,k}$ and $g_t$. We refer the interested reader to \cite{Goulart,sieber2021system, goulart2007output,zheng2019equivalence,furieri2019input} for classical and recent methods to overcome the non-convexity problem. For the rest of the paper, we assume that there exists a control input \eqref{eq:input} that complies with \eqref{eq:constraints} for all possible realizations of $w(t)$ and $v(t)$.

\begin{remark}
In this work, we analyze a  finite-horizon control problem, which represents one iteration of a receding-horizon MPC implementation. It is therefore appropriate to compare the proposed approach with a \emph{single} iteration of open-loop prediction approaches, such as the DeePC \cite{coulson2019data,berberich2020data}. The main difference is that we perform \emph{closed-loop predictions}, i.e., we optimize over feedback policies $\pi(\cdot)$ such that $u(t)=\pi(y(t),\dots,y(0))$, while the DeePC \cite{coulson2019data,berberich2020data} performs \emph{open-loop predictions}, i.e., it directly optimizes over input sequences  $u(0),u(1),u(N-1)$. It is well-known that closed-loop predictions are less conservative. Indeed, by setting $K_{t,k} = 0$ in \eqref{eq:input} the closed-loop policy reduces to an open-loop one. Most notably, closed-loop policies may preserve feasibility for significantly longer prediction horizons \cite{bemporad1998reducing}.  Naturally, the price to pay is an increased computational burden due to the larger dimensionality of the problem.  
\end{remark}

\subsection{Convex design through the IOP}
By leveraging tools offered by the framework of the  Input-Output Parametrization\footnote{Similar to \cite{goulart2007output,zheng2019equivalence}, the IOP \cite{furieri2019input} yields a convex representation of input-output closed-loop responses. It is also numerically stable for the case of infinite-horizon stable plants and for finite-horizon control problems \cite{zheng2019system}.} (IOP) \cite{furieri2019input}, one can formulate a convex optimization problem that computes the optimal safe feedback control policy by searching over the input-output closed-loop responses. The state-space equations \eqref{eq:dynamic} provide the following relations between trajectories
\begin{align}
    &\mathbf{x}_{[0,N-1]}=\mathbf{P}_{A}(:,0)x_0+\mathbf{P}_{B}\mathbf{u}_{[0,N-1]}\,,\label{eq:state_compact}\\ &\mathbf{y}_{[0,N-1]}=\mathbf{C}\mathbf{x}_{[0,N-1]}+\mathbf{v}_{[0,N-1]}\,,  \label{eq:output_compact}
\end{align}
where  $\mathbf{P}_A(:,0)$ denotes the first block-column of $\mathbf{P}_A$ and
\begin{alignat*}{3}
    &\mathbf{P}_{A}=(I-\mathbf{Z}\mathbf{A})^{-1}\,, &&\quad  \mathbf{P}_{B}=(I-\mathbf{Z}\mathbf{A})^{-1}\mathbf{Z}\mathbf{B}\,,\\
    & \mathbf{A}=I_{N} \otimes A\,, && \quad \mathbf{B}=I_{N}\otimes B\,,\\
    & \mathbf{C} = I_{N} \otimes C\,, && \quad \mathbf{Z}=\begin{bmatrix}
    0_{n \times n(N-1)}&0_{n \times n}\\
    I_{n(N-1)}&0_{n(N-1) \times n}
    \end{bmatrix}\,.
\end{alignat*}
A few comments on the used notation are in order. First, the matrix  $\mathbf{Z}$ is the block-downshift operator. Second, from now on we denote $\mathbf{G} = \mathbf{CP}_B$ to highlight that $\mathbf{G}$ is a block-Toeplitz matrix containing the first $N$ components of the impulse response of the plant $\mathbf{G}(z) = C(zI-A)^{-1}B$.  Last, the matrix $\mathbf{CP}_{A}(:,0)$ contains the entries of the observability matrix $CA^i$ for $i=0,\dots,N-1$. We denote the model-based free response of the system as $\mathbf{y}_{0}=\mathbf{CP}_{A}(:,0)x_0$. The control policy can be rewritten as:
\begin{equation}
    \label{eq:control_policy}
    \mathbf{u}_{[0,N-1]} = \mathbf{K}\mathbf{y}_{[0,N-1]}+ \mathbf{g}+ \mathbf{w}_{[0,N-1]}\,,
\end{equation}
where $\mathbf{K}$ and $\mathbf{g}$ are defined as:
\begin{equation}
\label{eq:K_sparsity}
{\small \mathbf{K} \hspace{-0.025cm} = \hspace{-0.075cm}
    \begin{bmatrix}
    K_{0,0}&0_{m \times p}&\dots&0_{m \times p}\\
    K_{1,0}&K_{1,1}&\ddots&0_{m \times p}\\
    \vdots&\vdots&\ddots&\vdots\\
    K_{N-1,0}&K_{N-1,1}&\dots&K_{N-1,N-1}
    \end{bmatrix}\hspace{-0.1cm},
    \mathbf{g} \hspace{-0.025cm} = \hspace{-0.075cm} \begin{bmatrix}g_0\\g_1\\\vdots\\ g_{N-1}\end{bmatrix}\hspace{-0.1cm}.}
\end{equation}
The safety constraints \eqref{eq:constraints}-\eqref{eq:polytope_definition} take the form
\begin{equation}
\label{eq:compact_constraints}
    \hspace{-0.4cm}\max_{\norm{\mathbf{v}}_\infty \leq v_\infty,~\norm{\mathbf{w}}_\infty \leq w_\infty} \hspace{-0.6cm}\mathbf{F}_y \mathbf{y} \leq \mathbf{b}_y\,,   \max_{\norm{\mathbf{v}}_\infty \leq v_\infty,\norm{\mathbf{w}}_\infty \leq w_\infty} \hspace{-0.6cm}\mathbf{F}_u \mathbf{u} \leq \mathbf{b}_u\,,
\end{equation}
with $\mathbf{F}_y \hspace{-0.25mm}=\hspace{-0.25mm} \operatorname{blkdiag}(F_y^0,\ldots,F_y^{N-1})$,
$\mathbf{b}_y \hspace{-0.25mm}=\hspace{-0.25mm} \operatorname{vec}(b_y^0,\ldots,b_y^{N-1})$,
$\mathbf{F}_u = \operatorname{blkdiag}(F_u^0,\ldots,F_u^{N-1})$,  $\mathbf{b}_u = \operatorname{vec}(b_u^0,\ldots,b_u^{N-1})$, and $\max(\cdot)$ to be intended row-wise.
By plugging the controller \eqref{eq:control_policy} into \eqref{eq:state_compact}-\eqref{eq:output_compact}, it is easy to derive the relationships
\begin{align}
    &\begin{bmatrix}\mathbf{y}\\ \mathbf{u}\end{bmatrix} =\begin{bmatrix}
           \bm{\Phi}_{yy} & \bm{\Phi}_{yu} \\
            \bm{\Phi}_{uy} & \bm{\Phi}_{uu} 
        \end{bmatrix}\begin{bmatrix} \mathbf{v}+\mathbf{y}_{0}\\ \mathbf{w}\end{bmatrix}+\begin{bmatrix}\mathbf{Gq}\\\mathbf{q}\end{bmatrix}\,, \label{eq:IOP_parameters}
\end{align}
where
\begin{equation}
   \bm{\Phi}\hspace{-0.12cm}=\hspace{-0.12cm}\begin{bmatrix}
           \bm{\Phi}_{yy} & \bm{\Phi}_{yu} \\
            \bm{\Phi}_{uy} & \bm{\Phi}_{uu} 
        \end{bmatrix}\hspace{-0.12cm} =\hspace{-0.12cm} \begin{bmatrix}(I\hspace{-0.1cm}-\hspace{-0.1cm}\mathbf{GK})^{-1} & (I\hspace{-0.1cm}-\hspace{-0.1cm}\mathbf{GK})^{-1}\mathbf{G}\\ \mathbf{K}(I-\mathbf{GK})^{-1} & (I-\mathbf{KG})^{-1}\end{bmatrix}\label{eq:CL_responses}\hspace{-0.1cm},
\end{equation}
and $\mathbf{q} = (I-\mathbf{KG})^{-1}\mathbf{g} = \bm{\Phi}_{uu}\mathbf{g}$. The parameters $(\bm{\Phi}_{yy}, \bm{\Phi}_{yu}, \bm{\Phi}_{uy}, \bm{\Phi}_{uu})$, {where $\bm{\Phi}_{yy}\in\mathbb{R}^{Np \times Np}, \bm{\Phi}_{yu}\in\mathbb{R}^{Np \times Nm}, \bm{\Phi}_{uy}\in\mathbb{R}^{Nm \times Np}$ and $\bm{\Phi}_{uu}\in\mathbb{R}^{Nm \times Nm}$}, represent the four closed-loop responses defining the relationship between disturbances and input-output signals, while $\mathbf{q} \in \mathbb{R}^{Nm}$ represents the affine part of the disturbance-feedback control policy \cite{goulart2007output,furieri2019unified}. To achieve a convex reformulation of the control problem under consideration, it is not hard to extend the IOP from \cite{furieri2019input} to account for the safety constraints \eqref{eq:compact_constraints} in a convex way. The result is summarized in the next proposition, whose proof is reported in Appendix \ref{app:prop:IOP} for completeness. 
\begin{proposition}
\label{prop:IOP}
Consider the LTI system \eqref{eq:dynamic} evolving under the control policy \eqref{eq:control_policy} within a horizon of length $N \in \mathbb{N}$. Then:
\begin{enumerate}[wide, labelwidth=!, labelindent=0pt]
    \item[$i)$] For any control policy $(\mathbf{K},\mathbf{g})$ that complies with the safety constraints, there exist four matrices ($\bm{\Phi}_{yy}, \bm{\Phi}_{yu}, \bm{\Phi}_{uy}, \bm{\Phi}_{uu}$) and a vector $\mathbf{q}$ such that $\mathbf{K} = \bm{\Phi}_{uy}\bm{\Phi}_{yy}^{-1}$, $\mathbf{g}= \bm{\Phi}_{uu}^{-1}\mathbf{q}$, and for all $j=1,\dots, sN$,
    \begin{align}
        &\begin{bmatrix} I & -\mathbf{G} \end{bmatrix}\bm{\Phi} = \begin{bmatrix} I & 0 \end{bmatrix},  \quad \bm{\Phi}\begin{bmatrix}  -\mathbf{G} \\I \end{bmatrix} = \begin{bmatrix} 0 \\ I\end{bmatrix},\label{eq:ach} \\
           &\norm{\begin{bmatrix}v_\infty(F_{y,j} \bm{\Phi}_{yy})^\mathsf{T}\\w_\infty(F_{y,j} \bm{\Phi}_{yu})^\mathsf{T}\end{bmatrix}^{\mathsf{T}}}^{\star}_1 \hspace{-0.1cm}+\hspace{-0.1cm} F_{y,j}(\mathbf{Gq}+ \bm{\Phi}_{yy}\mathbf{y}_{0})\leq\hspace{-0.05cm} \mathbf{b}_{y,j}\,, \label{eq:ach_constraints1}\\
            &\norm{\begin{bmatrix}v_\infty(F_{u,j}\bm{\Phi}_{uy})^\mathsf{T}\\w_\infty(F_{u,j} \bm{\Phi}_{uu})^\mathsf{T} \end{bmatrix}^{\mathsf{T}}}^{\star}_1 + F_{u,j}(\mathbf{q}+\bm{\Phi}_{uy}\mathbf{y}_{0}) \leq \mathbf{b}_{u,j} \,, \label{eq:ach_constraints2}\\
            &\bm{\Phi}_{yy}, \bm{\Phi}_{yu}, 
                \bm{\Phi}_{uy}, \bm{\Phi}_{uu} 
                \text{ with causal sparsities \footnotemark}\,, \label{eq:ach3}
    \end{align}
 where $F_{y,j} \in \mathbb{R}^{1 \times Np}$, $F_{u,j} \in \mathbb{R}^{1 \times Nm}$ and $\mathbf{b}_{u,j},\mathbf{b}_{y,j} \in \mathbb{R}$ are the $j$-th row of $\mathbf{F}_y$, $\mathbf{F}_u$ and $\mathbf{b}_{u},\mathbf{b}_y$, respectively.
    \footnotetext{Specifically, they have the block lower-triangular sparsities resulting as per the expressions \eqref{eq:CL_responses}, the sparsity of $\mathbf{K}$ in \eqref{eq:K_sparsity} and that  of $\mathbf{G}$.}
    \item[$ii)$] For any four matrices ($\bm{\Phi}_{yy}, \bm{\Phi}_{yu}, \bm{\Phi}_{uy}, \bm{\Phi}_{uu}$)  complying with \eqref{eq:ach}-\eqref{eq:ach3} and any vector $\mathbf{q} \in \mathbb{R}^{mN}$, the matrix $\mathbf{K}=\bm{\Phi}_{uy}\bm{\Phi}_{yy}^{-1}$ is causal as per \eqref{eq:K_sparsity} and it yields the closed-loop responses  ($\bm{\Phi}_{yy}, \bm{\Phi}_{yu}, \bm{\Phi}_{uy}, \bm{\Phi}_{uu}$). Moreover, the affine policy $(\mathbf{K},\mathbf{g}$) with $\mathbf{g} = \bm{\Phi}_{uu}^{-1}\mathbf{q}$ complies with the safety constraints.
\end{enumerate}
\end{proposition}

\vspace{0.2cm}

We remark that the IOP is well-suited to a data-driven output-feedback setup, as all affine control policies are directly parametrized through the impulse response parameters $\mathbf{G}$, without requiring an internal state-space representation. This is useful for two reasons. First, when dealing with unknown systems, the state-space parameters $(A,B,C,x_0)$ can only be estimated up to an unknown change of variable, which may be problematic for defining the cost and the noise statistics \cite{mania2019certainty}. Second, several large-scale systems feature a very large number of states, but a comparably small number of inputs and outputs, that is  $n>>\max(m,p)$. In such applications, it is advantageous to bypass a state-space representation and directly deal with $\mathbf{G}$, whose dimensions do not depend on $n$.

From now on, to simplify the expressions appearing throughout the next sections and without any loss of generality\footnote{One can redefine $\tilde{\mathbf{y}} = \begin{bmatrix}1&\mathbf{y}^\mathsf{T}\end{bmatrix}^\mathsf{T}$, $\overline{\mathbf{v}} = \begin{bmatrix}1&\mathbf{v}^\mathsf{T}\end{bmatrix}^\mathsf{T}$, $\overline{\mathbf{C}} = \begin{bmatrix}0_{1\times Nn}\\ \mathbf{C}\end{bmatrix}$, $\overline{\mathbf{K}} = \begin{bmatrix}\mathbf{g}&\mathbf{K}\end{bmatrix}$ and $\overline{\bm{\Phi}}$ as per \eqref{eq:CL_responses} with $\overline{\mathbf{G}}$ and  $\overline{\mathbf{K}}$ in place of $\mathbf{G}$ and $\mathbf{K}$, respectively. Minor modifications to \eqref{eq:ach_constraints1}-\eqref{eq:ach_constraints2} are needed as well.}, we let $\mathbf{q}=\mathbf{g} = 0_{Nm\times 1}$, that is, we focus on \emph{linear} control policies.  We are ready to establish a convex formulation of the optimal control problem under study. 
\begin{proposition}
\label{prop:strongly_convex}
Consider the LTI system \eqref{eq:dynamic}. The linear control policy that achieves the minimum of the cost functional \eqref{eq:cost_output} is given by $\mathbf{K} = \bm{\Phi}_{uy} \bm{\Phi}_{yy}^{-1}$, where $\bm{\Phi}_{uy},\bm{\Phi}_{yy}$ are optimal solutions to the following convex optimization problem:
\begin{align}
    &~\min_{\bm{\Phi}}\frobenius{
    \hspace{-0.025cm}
    \begin{bmatrix}\mathbf{Q}^{\frac{1}{2}}&0\\0&\mathbf{R}^{\frac{1}{2}}\end{bmatrix}
    \hspace{-0.1cm}
    \begin{bmatrix}
            \bm{\Phi}_{yy} & \bm{\Phi}_{yu} \\
            \bm{\Phi}_{uy} & \bm{\Phi}_{uu} 
        \end{bmatrix}
        \hspace{-0.1cm}
        \begin{bmatrix}\bm{\Sigma}^{\frac{1}{2}}_v&0&\mathbf{y}_{0}\\0&\bm{\Sigma}^{\frac{1}{2}}_w&0\end{bmatrix}
        \hspace{-0.025cm}
        }^2\label{prob:IOP}\\
    &\st~ \eqref{eq:ach}-\eqref{eq:ach3}\,, \nonumber
\end{align}
where $\mathbf{Q} = \text{blkdiag}(Q_0,\dots, Q_{N-1})$, $\mathbf{R} = \text{blkdiag}(R_0,\dots, R_{N-1})$, $\bm{\Sigma}_v = I_{N} \otimes \Sigma_v$, $\bm{\Sigma}_w = I_{N} \otimes \Sigma_w$ and where \eqref{eq:ach_constraints1}-\eqref{eq:ach_constraints2} are evaluated at $\mathbf{q}=0$.
\end{proposition}
\begin{proof}
We refer to Proposition~2 of \cite{furieri2021behavioral} for a complete derivation of the cost function. To conclude the proof, it suffices to notice that the objective function and the safety constraints \eqref{eq:ach_constraints1}-\eqref{eq:ach_constraints2} are convex in $\bm{\Phi}$.
\end{proof}

\vspace{0.2cm}

When the system parameters $(A,B,C,x_0)$ are known, a globally optimal solution $(\bm{\Phi}^\star_{yy}, \bm{\Phi}^\star_{yu}, \bm{\Phi}^\star_{uy}, \bm{\Phi}^\star_{uu}$) for problem~\eqref{prob:IOP} can be efficiently computed with off-the-shelf solvers. The corresponding globally optimal and safe control policy is then recovered as $\mathbf{K}^\star = \bm{\Phi}_{uy}^\star (\bm{\Phi}^\star_{yy})^{-1}$.

\vspace{0.2cm}
The rest of the paper contains our main contributions. Specifically, we address  the following two questions:
\begin{itemize}
    \item[Q1) ] How can we compute a safe control policy with performance close to that of $\mathbf{K}^\star$, solely based on libraries of \emph{noisy} input-output trajectories? 
    \item[Q2) ] How steeply does the suboptimality grow with respect to $\mathbf{K}^\star$ as the noise increases?
\end{itemize}

\section{The Data-Driven Case: Robustly Safe Controller Synthesis From Noisy Data}
\label{sec:Noisy_BIOP_constrained}

We answer question Q1) by developing a method to synthesize near-optimal safe controllers from noisy data. The main result of this section is an optimization problem based on the IOP that tightly approximates the optimal and safe control policy, despite the fact that the noise-corrupted data only yield approximate estimates of the system impulse and free response. We conclude by offering novel insights on its properties and its numerical implementation based on convex optimization.

\subsection{From noise-corrupted data to doubly-robust optimal control}

From now on, the dynamics matrices $(A,B,C)$ and the initial state $x_0$ are \emph{unknown}. Instead, only the following data are available:
\begin{enumerate}
    \itemsep0em 
    \item[$\mathbf{D1}$] A noisy system trajectory $\{y^h(t),u^h(t)\}_{t=-T}^{-1}$ recorded offline during an experiment.
    \item[$\mathbf{D2}$] The cost matrices $Q_t$, $R_t$, the matrices $\Sigma_v, \Sigma_w$, the safety sets $\Gamma_t$, and the bounded sets $\bm{\mathcal{W}} = \{\mathbf{w}|~\norm{\mathbf{w}}_\infty \leq w_\infty\}$ and  $\bm{\mathcal{V}} = \{\mathbf{v}|~\norm{\mathbf{v}}_\infty \leq v_\infty\}$ where disturbances live.
\end{enumerate}

Our approach exploits the noisy data in $\mathbf{D1}$ to compute approximate system responses $\hatbf{G}$ and $\hatbf{y}_0$ in a preliminary identification step. We work under the following assumption. 
\begin{assumption}
\label{ass:bounded_error}
Let $\bm{\Delta} = \mathbf{G}-\hatbf{G}$ and $\bm{\delta}_0= \mathbf{y}_{0}-\hatbf{y}_{0}$. There exist $\epsilon_{2,G},\epsilon_{\infty,G},\epsilon_{2,y},\epsilon_{\infty,y}>0$ such that, \begin{alignat*}{3}
&\normtwo{\bm{\Delta}} \leq \epsilon_{2,G}, ~~~ &&\normtwo{\bm{\delta}_0} \leq \epsilon_{2,y}\,,\\
&\norm{\bm{\Delta}}_\infty \leq \epsilon_{\infty,G},  && \norm{\bm{\delta}_{0}}_\infty \leq \epsilon_{\infty,y}\,.
\end{alignat*}

\end{assumption}

\smallskip
Note that, in practice, a meaningful bound on $\delta_0$ is only available if $A$ is stable or the time-horizon is sufficiently short. Let us define $\epsilon_2 = \max(\epsilon_{2,G},\epsilon_{2,y})$ and $\epsilon_{\infty} = \max(\epsilon_{\infty,G},\epsilon_{\infty,y})$.
Assumption~\ref{ass:bounded_error} can be fulfilled using different methods over the available data $\mathbf{D1}$; for instance, one may utilize standard least-squares identification that comes with probabilistic and non-asymptotic error bounds \cite{oymak2019non,zheng2020non}, or more sophisticated stochastic estimators based on behavioral theory such as maximum-likelihood predictors \cite{yin2021maximum}, which also come with quantifiable error bounds \cite{yin2021data}. Our results are independent of the choice of the identification scheme. A discussion as to how recent behavioral approaches can be used for identification is reported in Appendix~\ref{app:Willems}. These estimators will be used in the numerical examples in Section~\ref{sec:simulations}.

After condensing the effect of noise-corrupted data into model mismatch parameters $\bm{\Delta},\bm{\delta}_0$, we formulate a \emph{doubly-robust}
control problem, that is, a problem where we enforce constraint satisfaction for $1)$ all possible model mismatches ($\bm{\Delta},\bm{\delta}_0)$, and $2)$ all possible disturbances sequences $\mathbf{w}\in \bm{\mathcal{W}}$ and $\mathbf{v} \in \bm{\mathcal{V}}$. In particular, define $\bm{\theta} = (\bm{\Delta},\bm{\delta}_0,\mathbf{w},\mathbf{v})$ and let
\begin{align*}
    \mathbf{y}(\mathbf{K},\bm{\theta}) &= \hatbf{y}_0+\bm{\delta}_0+(\hatbf{G}+\bm{\Delta})\mathbf{u}(\mathbf{K},\bm{\theta})+\mathbf{v}\,,\\
    \mathbf{u}(\mathbf{K},\bm{\theta})&=\mathbf{K}\mathbf{y}(\mathbf{K},\bm{\theta}) +\mathbf{w}\,,
\end{align*}
be the closed-loop trajectories associated with a specific controller $\mathbf{K}$ and disturbance and mismatch realizations $\bm{\theta}$. Further, define the set of doubly-robust controllers as:
\begin{equation*}
    \bm{\mathcal{K}}=\{\mathbf{K}\text{ in \eqref{eq:K_sparsity}}|~(\mathbf{y}(\mathbf{K},\bm{\theta}),\mathbf{u}(\mathbf{K},\bm{\theta})) \in \bm{\Gamma},~\forall \bm{\theta} \in \bm{\mathcal{E}}\times \bm{\mathcal{W}} \times \bm{\mathcal{V}}\}\,,
\end{equation*}
with $\bm{\mathcal{E}} = \{ (\bm{\Delta},\bm{\delta}_0)|\norm{\bm{\Delta}}_p \leq \epsilon_{p},~\norm{\bm{\delta}_0}_p \leq \epsilon_{p}\text{, $\forall p\in \{2,\infty\}$}\}$, $\bm{\Gamma} = \Gamma_0 \times \Gamma_1 \times \cdots \times \Gamma_{N-1}$, and assume that $\bm{\mathcal{K}}$ is not empty. Then, the doubly-robust problem of interest takes the form
\begin{alignat}{3}
\min_{\mathbf{K} \in \bm{\mathcal{K}}}\max_{(\bm{\Delta,\bm{\delta}_0}) \in \bm{\mathcal{E}}}\sqrt{\hspace{-0.1cm}\underset{\mathbf{w},\mathbf{v}}{\mathbb{E}}\left[\mathbf{y}(\mathbf{K},\hspace{-0.05cm}\bm{\theta})^\mathsf{T}\hspace{-0.05cm}\mathbf{y}(\mathbf{K},\hspace{-0.05cm}\bm{\theta})\hspace{-0.05cm}+\hspace{-0.05cm}\mathbf{u}(\mathbf{K},\hspace{-0.05cm}\bm{\theta})^\mathsf{T}\mathbf{u}(\mathbf{K},\hspace{-0.05cm}\bm{\theta})\right]}\,,\label{prob:robust_deltas}
\end{alignat}
where we have selected the weights $\mathbf{Q}$, $\mathbf{R}$, $\bm{\Sigma}_w$, $\bm{\Sigma}_v$ to be identity matrices with appropriate dimensions. The same assumption is used in the rest of the paper, in order to facilitate the derivations. However, we note that all our results can be easily adapted to non-identity weights. Next, we observe that the doubly-robust optimization problem admits an equivalent formulation in terms of the closed-loop response parameters.

\begin{proposition}
\label{pr:robust_IOP}
Letting $\bm{\Phi}_{yy} = \hatbm{\Phi}_{yy}(I-\mathbf{\Delta}\hatbm{\Phi}_{uy})^{-1},~\bm{\Phi}_{yu} = \bm{\Phi}_{yy}(\hatbf{G}+\mathbf{\Delta}),~
\bm{\Phi}_{uy} = \hatbm{\Phi}_{uy}(I-\mathbf{\Delta}\hatbm{\Phi}_{uy})^{-1},~\bm{\Phi}_{uu} = (I-\hatbm{\Phi}_{uy}\mathbf{\Delta})^{-1}\hatbm{\Phi}_{uu}$, the optimization problem \eqref{prob:robust_deltas} is equivalent to
\begin{align}
\min_{\hatbm{\Phi} \in \bm{\Pi}} \max_{(\bm{\Delta,\bm{\delta}_0}) \in \bm{\mathcal{E}}}\norm{\begin{bmatrix}
           \bm{\Phi}_{yy}&\bm{\Phi}_{yu}\\\bm{\Phi}_{uy}&\bm{\Phi}_{uu}\end{bmatrix} \Bigg[\begin{matrix}I&0&\hatbf{y}_{0}+\bm{\delta}_0 \\0&I&0\end{matrix}\Bigg]  }_F\label{prob:robust_IOP}\,,
\end{align}
where the set of doubly-robust closed-loop responses $\bm{\Pi}$ is
\begin{equation*}
        \bm{\Pi}=\{\hatbm{\Phi}|~\eqref{eq:constraints_hatPhi1}-\eqref{eq:constraints_hatPhi3}, ~ \forall j=1,\dots,sN,~\forall (\bm{\Delta},\bm{\delta}_0) \in \bm{\mathcal{E}} \}\,,
\end{equation*}
with 
\begin{align}
        &\begin{bmatrix}
             I&-\hatbf{G}
             \end{bmatrix}\hatbm{\Phi}=\begin{bmatrix}
             I&0
             \end{bmatrix}, \quad \hatbm{\Phi} 
             \begin{bmatrix}
             -\hatbf{G}\\I
             \end{bmatrix}=\begin{bmatrix}
             0\\I
             \end{bmatrix}, \label{eq:constraints_hatPhi1} \\
        &  \norm{\begin{bmatrix}v_\infty\left(F_{y,j} \bm{\Phi}_{yy}\right)^\mathsf{T}\\w_\infty\left(F_{y,j} \bm{\Phi}_{yu}\right)^\mathsf{T}\end{bmatrix}}_1\hspace{-0.2cm}+\hspace{-0.1cm}  \left(F_{y,j} \bm{\Phi}_{yy}\right)(\hatbf{y}_{0}+\bm{\delta}_0)\leq \mathbf{b}_{y,j}\,, \label{eq:constraints_delta1}\\
            &  \norm{\begin{bmatrix}v_\infty\left(F_{u,j}\bm{\Phi}_{uy}\right)^\mathsf{T}\\w_\infty\left(F_{u,j} \bm{\Phi}_{uu}\right)^\mathsf{T} \end{bmatrix}}_1 \hspace{-0.2cm}+\hspace{-0.1cm}  \left(F_{u,j} \bm{\Phi}_{uy}\right)(\hatbf{y}_{0}\hspace{-0.1cm}+\hspace{-0.1cm}\bm{\delta}_0) \leq \mathbf{b}_{u,j} \label{eq:constraints_delta2}\,,\\
            &    \hatbm{\Phi}_{yy}, \hspace{-0.05cm}\hatbm{\Phi}_{yu}, \hspace{-0.05cm}\hatbm{\Phi}_{uy}, \hspace{-0.05cm} \hatbm{\Phi}_{uu} \text{ with causal sparsities.}    \label{eq:constraints_hatPhi3}
    \end{align}
\end{proposition}
\color{black}
The proof of Proposition~\ref{pr:robust_IOP} can be found in Appendix \ref{app:prop:robust_IOP}. We remark that the closed-loop responses $\bm{\Phi}$ appearing in \eqref{prob:robust_IOP}, \eqref{eq:constraints_delta1} and \eqref{eq:constraints_delta2} are associated with the $\emph{true}$ impulse response, whereas the closed-loop responses $\hatbm{\Phi}$ appearing in \eqref{eq:constraints_hatPhi1} and \eqref{eq:constraints_hatPhi3} are associated with the \emph{estimated} impulse response. This is because, while we are interested in minimizing the cost and satisfying the safety constraints for the \emph{real} system, we can only parametrize the closed-loop responses for the identified system.

The robust optimization problem \eqref{prob:robust_IOP} is non-convex in the cost and in the constraints because $\bm{\Phi}$ is a nonlinear function of the  matrix variables $\hatbm{\Phi}$ and $\bm{\Delta}$. Therefore, it is challenging to find a feasible solution, let alone the optimal one. We note that, for the case of open-loop control policies, one may use constraint-tightening approaches such as those of \cite{tanaskovic2014adaptive,terzi2019learning}. In this work, we propose an analysis that compares feedback control policies. Specifically, we derive suboptimality guarantees with respect to the optimal model-based linear feedback policy as a function of the model mismatch level.

\subsection{Proposed relaxation for safe controller synthesis}
\label{sub:proposed_reformulation}
Our first main result is to derive a relaxation of the intractable problem \eqref{prob:robust_IOP} that we can solve in practice.
Our proposed approach is to 1) upper bound the cost function, and 2) tighten the safety constraints with more tractable expressions. In Section~\ref{sec:suboptimality} we will explicitly quantify the suboptimality incurred by these approximations. At its core, this methodology is inspired by that developed in \cite{dean2019safely} for the state-feedback case without measurement noise. However, the addition of output-feedback and measurement noise leads to new terms both in the cost and the safety constraints that are more challenging to analyze. 

The following two lemmas establish the basis for our relaxation. Let $J(\mathbf{G},\mathbf{K}) = \sqrt{\mathbb{E}_{\mathbf{w},\mathbf{v}}\left[\mathbf{y}^\mathsf{T}\mathbf{y}+\mathbf{u}^\mathsf{T}\mathbf{u}\right]}$ denote the square root of the cost in \eqref{eq:cost_output}. Lemma~\ref{le:upperbound} provides the new expression which upper bounds $J(\mathbf{G},\mathbf{K})$ and Lemma~\ref{le:quasiconvex_safetyconstraints} provides a tightened form of the safety constraints. Their rather lengthy technical proof is reported in the Appendices \ref{app:le:upperbound} and \ref{app:le:quasiconvex_safetyconstraints}, respectively. 
\begin{lemma}
\label{le:upperbound}
Let $\hatbm{\Phi}$ denote the closed-loop responses obtained by applying $\mathbf{K}$ to $\hatbf{G}$. Further assume that $\norm{\hatbm{\Phi}_{uy}}_2 \leq \gamma$, where $\gamma \in [0,\epsilon_2^{-1})$.  Then, we have
\begin{align}
     \label{eq:nonconvex_bound}
    & J(\mathbf{G},\mathbf{K}) \leq  \frac{J_{UB}}{1-\epsilon_2 \gamma} 
 \end{align}
 where 
 {\small
 \begin{equation*}
     \hspace{-0.08cm}J_{UB}= \hspace{-0.08cm}\norm{\begin{bmatrix}
              \sqrt{1\hspace{-0.08cm}+\hspace{-0.08cm}h(\epsilon_2,\gamma,\hatbf{G}) \hspace{-0.08cm}+\hspace{-0.08cm} h(\epsilon_2,\gamma,\hatbf{y}_0)}\hatbm{\Phi}_{yy} & \hatbm{\Phi}_{yu} & \hatbm{\Phi}_{yy}\hatbf{y}_0\\\sqrt{1+ h(\epsilon_2,\gamma,\hatbf{y}_0)}\hatbm{\Phi}_{uy} & \hatbm{\Phi}_{uu} & \hatbm{\Phi}_{uy}\hatbf{y}_0
             \end{bmatrix}}_F \hspace{-0.2cm},
 \end{equation*}
 }
and $h(\epsilon, \gamma,\mathbf{Y}) = \epsilon^2(2 + \gamma\|\mathbf{Y}\|_2 )^2+2\epsilon\norm{\mathbf{Y}}_{2}(2+\gamma\norm{\mathbf{Y}}_2)$.
\end{lemma}

Lemma~\ref{le:upperbound} exploits the upper bound $\norm{\hatbm{\Phi}_{uy}}_2\leq \gamma$ to establish an explicit relationship between  $J(\mathbf{G},\mathbf{K})$, the cost obtained by applying a controller $\mathbf{K}$ to the real system $\mathbf{G}$, and $J(\hatbf{G},\mathbf{K})$, the cost obtained by applying the same controller to the estimated system $\hatbf{G}$. To see this, notice that \eqref{eq:nonconvex_bound} can be equivalently rewritten as
\begin{align}
    J(\mathbf{G},\mathbf{K})   &\leq \frac{\Big(J(\hatbf{G},\mathbf{K})^2 +
    \|\hatbm{\Phi}_{yy}\|_F^2 (h(\epsilon_2,\gamma,\hatbf{G})+}{1-\epsilon_2 \gamma}\nonumber\\&  \frac{+
    h(\epsilon_2,\gamma,\hatbf{y}_0)) + 
    \|\hatbm{\Phi}_{uy}\|_F^2 h(\epsilon_2,\gamma,\hatbf{y}_0)\Big)^{\frac{1}{2}}}{1-\epsilon_2 \gamma}\,.\label{eq:upper_bound_reformulated}
\end{align}
The expression \eqref{eq:upper_bound_reformulated} upper bounds the gap between $J(\mathbf{G},\mathbf{K})$ and $J(\hatbf{G},\mathbf{K})$ as a quantity that increases with $\epsilon_2$ and with the norm of  $\hatbf{G},\hatbf{y}_0,\hatbm{\Phi}$. We note that a similar result has appeared in \cite[Proposition 3.2]{zheng2020sample}. However, Lemma~\ref{le:upperbound} additionally takes into account how an uncertain $\hat{x}(0)$ affects the cost through the free response $\hatbf{y}_0$. 
We now derive a tightened - yet more tractable - expression for the safety constraints \eqref{eq:constraints_delta1}-\eqref{eq:constraints_delta2}.
\begin{lemma}
\label{le:quasiconvex_safetyconstraints}
Assume $\norm{\hatbm{\Phi}_{uy}}_\infty \leq \tau$, where $\tau \in [0,\epsilon_\infty^{-1})$. Then, if for all $j=1,\dots,sN$ the closed-loop responses $\hatbm{\Phi}$ satisfy the tightened safety constraints
 \begin{align}
     &f_{1,j}(\hatbm{\Phi})+f_{2,j}(\hatbm{\Phi})+f_{3,j}(\hatbm{\Phi}) \leq \mathbf{b}_{y,j}\,,\label{eq:quasiconvex_constraints1}\\
     &f_{4,j}(\hatbm{\Phi})+f_{5,j}(\hatbm{\Phi})+f_{6,j}(\hatbm{\Phi}) \leq \mathbf{b}_{u,j}\label{eq:quasiconvex_constraints2}\,,
 \end{align}
where
 \begin{alignat*}{3}
     &f_{1,j}(\hatbm{\Phi}) = \frac{v_\infty \norm{F_{y,j} \hatbm{\Phi}_{yy}}^{\star}_1}{1-\epsilon_\infty \tau}\,, \quad f_{4,j}(\hatbm{\Phi}) = \frac{v_\infty \norm{F_{u,j} \hatbm{\Phi}_{uy}}^{\star}_1}{1-\epsilon_\infty \tau}\,,\\
     &f_{2,j}(\hatbm{\Phi}) = w_\infty\norm{\begin{bmatrix}\left(F_{y,j} \hatbm{\Phi}_{yu}\right)^\mathsf{T}\\\epsilon_\infty\frac{1+\tau \norm{\hatbf{G}}_\infty}{1-\epsilon_\infty \tau} \left(F_{y,j}\hatbm{\Phi}_{yy}\right)^\mathsf{T}\end{bmatrix}}_1\,,\\
     &f_{5,j}(\hatbm{\Phi}) = w_\infty \norm{\begin{bmatrix}\left(F_{u,j}\hatbm{\Phi}_{uu}\right)^\mathsf{T}\\\epsilon_\infty\frac{1+\tau \norm{\hatbf{G}}_\infty}{1-\epsilon_\infty \tau}\left(F_{u,j}\hatbm{\Phi}_{uy}\right)^\mathsf{T}\end{bmatrix}}_1\,,\\
     &f_{3,j}(\hatbm{\Phi}) = F_{y,j} \hatbm{\Phi}_{yy} \hatbm{y}_0+ \epsilon_\infty \hspace{-0.1cm}\norm{F_{y,j} \hatbm{\Phi}_{yy}}^{\star}_1
     \left(
     \frac{1+\tau \norm{\hatbf{y}_0}_\infty}{1-\epsilon_\infty \tau}
     \right)\,,\\
     &f_{6,j}(\hatbm{\Phi}) = F_{u,j} \hatbm{\Phi}_{uy} \hatbm{y}_0+ \epsilon_\infty\hspace{-0.1cm} \norm{F_{u,j} \hatbm{\Phi}_{uy}}^{\star}_1
     \left(
     \frac{1+\tau \norm{\hatbf{y}_0}_\infty}{1-\epsilon_\infty \tau}\right)\,,
  \end{alignat*}
  then $\hatbm{\Phi}$ satisfies the safety constraints \eqref{eq:constraints_delta1}-\eqref{eq:constraints_delta2} for all $(\bm{\Delta},\bm{\delta}_0) \in \bm{\mathcal{E}}$. 
\end{lemma}

Lemma~\ref{le:quasiconvex_safetyconstraints} exploits the upper bound $\norm{\hatbm{\Phi}_{uy}}_\infty\leq \tau$ to quantify the worst-case effect of the disturbances in increasing the values of the inputs and the outputs. In our setup, similar to \cite{dean2019safely}, the feasible set shrinks in the presence of larger impulse and free response estimation error $\epsilon_\infty$. This is because \eqref{eq:quasiconvex_constraints1}-\eqref{eq:quasiconvex_constraints2} are more restrictive, and will eventually become infeasible for sufficiently large $\epsilon_\infty$. Instead, the effect of increasing the value of $\tau$ is less intuitive. Indeed,  as $\tau$ increases, the constraint $\norm{{\hatbf{\Phi}}_{uy}}_\infty \leq \tau$ softens while \eqref{eq:quasiconvex_constraints1}-\eqref{eq:quasiconvex_constraints2} tighten. It is therefore necessary to explicitly optimize over $\tau$. We are now ready to establish a  relaxation of problem~\eqref{prob:robust_IOP}. 

\begin{theorem}
\label{thm: robust formulation}
Consider the following optimization problem: 
\begin{alignat}{3}
&\min_{\gamma \in [0,\epsilon_2^{-1}), \tau \in [0,\epsilon_\infty^{-1})}&& \frac{1}{1-\epsilon_2 \gamma}  \min_{\hatbm{\Phi}} \qquad J_{UB} \label{prob:quasi_convex}\\
&\st~&&\eqref{eq:constraints_hatPhi1},~\eqref{eq:constraints_hatPhi3}\,,\nonumber \\
        &~~&& \norm{\hatbm{\Phi}_{uy}}_2 \leq \gamma\,, \quad \norm{\hatbm{\Phi}_{uy}}_\infty \leq \tau\,,\label{eq:norm_constraints_quasiconvex}\\
     &~&& \eqref{eq:quasiconvex_constraints1}-\eqref{eq:quasiconvex_constraints2}, \quad \forall j=1,\dots,sN \nonumber\,,
\end{alignat}
where  $J_{UB}$ is defined in Lemma~\ref{le:upperbound}. Then,  \eqref{prob:quasi_convex} has the following properties:
\begin{enumerate}
    \item[$i)$] upon fixing any specific values for $\gamma \in [0,\epsilon_2^{-1})$ and $\tau \in [0,\epsilon_\infty^{-1})$, the optimization problem is convex in $\hatbm{\Phi}$,
    \item[$ii)$] all of its feasible solutions yield a controller $\hatbf{K} = \hatbm{\Phi}_{uy}\hatbm{\Phi}_{yy}^{-1}$ complying with the safety constraints \eqref{eq:ach_constraints1}-\eqref{eq:ach_constraints2} for the real system,
    \item[$iii)$] its minimal cost upper bounds that of \eqref{prob:robust_deltas}.
\end{enumerate}
\end{theorem}
\begin{proof}
Lemma~\ref{le:upperbound}  shows that the cost of \eqref{prob:quasi_convex} upper bounds $J(\mathbf{G},\mathbf{K}) = J(\mathbf{G},\hatbm{\Phi}_{uy}\hatbm{\Phi}_{yy}^{-1})$ for every feasible $\mathbf{K}$. Lemma~\ref{le:quasiconvex_safetyconstraints} shows that \eqref{eq:quasiconvex_constraints1}-\eqref{eq:quasiconvex_constraints2} imply the doubly-robust constraints \eqref{eq:constraints_delta1}-\eqref{eq:constraints_delta2} for all $(\bm{\Delta},\bm{\delta}_0) \in \bm{\mathcal{E}}$. Hence, $\hatbf{K} = \hatbm{\Phi}_{uy}\hatbm{\Phi}_{yy}^{-1}$ complies with safety constraints \eqref{eq:ach_constraints1}-\eqref{eq:ach_constraints2} for the real system. When $\gamma$ and $\tau$ are fixed, it remains to optimize over $\hatbm{\Phi}$. The cost function is convex in $\hatbm{\Phi}$ and so are the constraints of the inner optimization problem.
\end{proof}

Theorem~\ref{thm: robust formulation} shows that problem~\eqref{prob:robust_IOP}, which is non-convex in its matrix variables, can be approximated as the problem of  solving a convex optimization problem\footnote{Specifically, a semidefinite program (SDP) due to the presence of quadratic $\norm{\cdot}_2$ constraints.} for each choice of the  scalar variables $\gamma$ and $\tau$. The $\epsilon$-dependent suboptimality introduced by such an approximation will be quantified in the next section. The global optimum of \eqref{prob:quasi_convex} is thus determined by exhaustive search over the box $(\gamma,\tau) \in [0,\epsilon_2^{-1})\times [0,\epsilon_\infty^{-1})$, for instance through gridding, random search \cite{bergstra2012random} or bisection \cite{chen2020robust}. Gridding over $(\tau,\gamma)$ and solving a convex optimization problem each time may significantly increase the computational burden if we are interested in determining a near-optimal solution with very low tolerance. Similarly to \cite{zheng2020sample}, in the next proposition we show that the inner cost function in problem \eqref{prob:quasi_convex} can be made independent of $\gamma$ by introducing a parameter $\alpha \in \mathbb{R}$ that acts as an upper bound to $\gamma$. As a result, the overall cost becomes quasiconvex\footnote{A function $f:\mathbb{R}^n\rightarrow \mathbb{R}$ is quasiconvex if and only if $f(\theta x_1+(1-\theta)x_2)\leq \max(f(x_1),f(x_2))$ for every $x_1,x_2 \in \mathbb{R}^n$ and every $\theta \in [0,1]$. We refer to  \cite{agrawal2020disciplined} for a comprehensive discussion.} in $\gamma$, and the globally optimal $\gamma^\star(\tau)$ for each fixed $\tau$ can be found efficiently through golden-section search \cite{kiefer1953sequential}.
\begin{proposition}
\label{pr:alpha}
Fix $ \alpha \in [0, \epsilon_2^{-1})$ and consider the following  optimization problem
\begin{alignat}{3}
&\min_{\gamma \in [0,\alpha], \tau \in [0,\epsilon_\infty^{-1})}&& \frac{1}{1-\epsilon_2 \gamma}  \min_{\hatbm{\Phi}} \qquad J_{UB}^\alpha(\hatbm{\Phi}) \label{prob:quasi_convex_alpha}\\
&\st~&&\eqref{eq:constraints_hatPhi1},~\eqref{eq:constraints_hatPhi3},~\eqref{eq:quasiconvex_constraints1},~ \eqref{eq:quasiconvex_constraints2},~\eqref{eq:norm_constraints_quasiconvex} \nonumber \\
     &~&& \forall j=1,\dots,sN \nonumber\,,
\end{alignat}
where  $J_{UB}^\alpha(\hatbm{\Phi})$ is defined as
\begin{equation*}
    \norm{\begin{bmatrix}
              \sqrt{1\hspace{-0.08cm} +\hspace{-0.08cm} h(\epsilon_2,\alpha,\hatbf{G})\hspace{-0.08cm} +\hspace{-0.08cm}  h(\epsilon_2,\alpha,\hatbf{y}_0)}\hatbm{\Phi}_{yy} & \hatbm{\Phi}_{yu} & \hatbm{\Phi}_{yy}\hatbf{y}_0\\\sqrt{1+ h(\epsilon_2,\alpha,\hatbf{y}_0)}\hatbm{\Phi}_{uy} & \hatbm{\Phi}_{uu} & \hatbm{\Phi}_{uy}\hatbf{y}_0
             \end{bmatrix}}_F\hspace{-0.15cm} .
\end{equation*}
Then, the statements $i)$, $ii)$ and $iii)$ of Theorem~\ref{thm: robust formulation} hold. Furthermore
\begin{itemize}
    \item [$iv)$] The cost function of problem \eqref{prob:quasi_convex_alpha} is quasiconvex in $\gamma$. 
\end{itemize}
\end{proposition}
\begin{proof}
Since $\gamma \leq \alpha$, $\alpha< \epsilon_2^{-1}$, and $h(\epsilon,\gamma,\cdot)$ is a monotonically increasing function of $\gamma$, then the inequality $\eqref{eq:nonconvex_bound}$ in Lemma~\ref{le:upperbound} continues to hold when putting $\alpha$ in place of $\gamma$ inside the $h(\cdot)$ functions. The constraints of \eqref{prob:quasi_convex} are unaffected. Hence, $i)$, $ii)$ and $iii)$ of Theorem~\ref{thm: robust formulation} continue to hold. It remains to prove $iv)$. Let us fix any value for $\tau$. First, notice that $J_{UB}^\alpha(\hatbm{\Phi})$ is a convex function of $\hatbm{\Phi}$ and does not depend on $\gamma$, and that the feasible set of the inner minimization in problem \eqref{prob:quasi_convex_alpha} is convex. Denote as $g(\gamma)$ the optimal value of the inner optimization problem. We are left with minimizing the functional $\frac{g(\gamma)}{1-\epsilon_2 \gamma}$ over $\gamma$. We know that $g(\gamma)$ is convex in $\gamma$ because it is obtained as the partial minimization of a convex functional over a convex set \cite{boyd2004convex}, and that $(1-\epsilon_2\gamma)$ is concave in $\gamma$. Since the ratio of a non-negative convex function and a positive concave function is quasiconvex, we conclude that the cost of problem~\eqref{prob:quasi_convex_alpha} is quasiconvex in $\gamma$.
\end{proof}

In \cite{zheng2020sample}, the idea of using the parameter $\alpha$ was relying on a lemma from \cite{matni2017scalable}. Here, we have derived an alternative self-contained proof that holds also for the case $x_0\neq 0$. In summary, for a fixed $\alpha < \epsilon_2^{-1}$, for every $\tau$ gridding the interval $[0,\epsilon_\infty^{-1}]$ and for $\gamma$ chosen according to golden-search, we solve the corresponding instance of the inner optimization problem in \eqref{prob:quasi_convex_alpha}, which is convex in $\hatbm{\Phi}$. We also note that an infinite-horizon version of problem~\eqref{prob:quasi_convex_alpha} can be established by adding a tail variable and adopting a finite-horizon approximation of stable transfer functions similar to \cite{dean2019safely}.

Last, one may wonder whether the cost function of problem~\eqref{prob:quasi_convex_alpha} is jointly quasiconvex in $\gamma$ and $\tau$, as conjectured in \cite{dean2019safely}. Here, we clarify that this may not be the case, even for the state-feedback framework of \cite{dean2019safely}. For instance, similar to the constraints \eqref{eq:quasiconvex_constraints1}-\eqref{eq:quasiconvex_constraints2} and those of \cite{dean2019safely}, consider the function $s:\mathbb{R}^2 \rightarrow \mathbb{R}$ defined as $s(x,y) = \frac{|y|x}{(1-x)}$. Fixing $x_1=1,x_2=0,y_1=-\frac{1}{2},y_2=\frac{1}{2}$, one can verify that $$s(\theta x_1+(1-\theta)x_2,\theta y_1 +(1-\theta)y_2) = |0.5-\theta| \frac{\theta}{1-\theta}\,,$$ is not quasiconvex for $\theta \in [0,1]$. Based on this reasoning, the cost of problem~\eqref{prob:quasi_convex_alpha}, and similarly the objective (2.3) in \cite{dean2019safely}, may not be quasiconvex in $\tau$. Hence, exhaustive search over $\tau$ remains the only solution in general. Table~\ref{table:convex_prop} summarizes the convexity properties of \eqref{prob:quasi_convex} and \eqref{prob:quasi_convex_alpha}.

\begin{table}
    \begin{center}
        \caption{Convexity properties for the proposed reformulations.}
        \label{tab:comparison_MLP}
        \small
        \begin{tabular}{c|c|c|c}
            &QC  in $\gamma$& QC in ($\tau$,$\gamma$) &C for fixed $(\gamma,\tau)$  \\
            \hline
            \eqref{prob:quasi_convex} & X &  X & $\checkmark$\\ 
            \hline
            \eqref{prob:quasi_convex_alpha} & $\checkmark$ & X&$\checkmark$ \\
            \hline
        \end{tabular}
            \label{table:convex_prop}
    \end{center}
\end{table}

\subsection{Safe exploration}
In many applications, it is desirable not only that the control policy synthesized from data is safe, but also that the system operates safely during the data-collection phase. In our setup, this amounts to requiring that the available trajectories $\mathbf{D1}$ verify \eqref{eq:constraints}. Similar to \cite{dean2019safely}, we now clarify that safe data can be collected  exploiting the feasible space of the optimization problem \eqref{prob:quasi_convex} as a corollary of Theorem~\ref{thm: robust formulation}.

More in details, assume that rough estimates $\hatbf{G}_{r},\hatbf{y}_{0,r}$ are given with possibly large errors $\epsilon_{\infty,r}$ and $\epsilon_{2,r}$. Note that it is inherently impossible to guarantee safe exploration unless some prior information is available. Consider now an exploration signal $\eta(t)$ such that $\norm{\eta(t)}_\infty\leq \eta_\infty$ for every $t \in \mathbb{Z}$ and define $\widetilde{w}_\infty = w_\infty+\eta_\infty$. Let $\mathbf{K}_r = \hatbm{\Phi}_{uy,r}\hatbm{\Phi}_{yy,r}^{-1}$ be any feasible solution to the instance of the optimization problem \eqref{prob:quasi_convex} where we use $\widetilde{w}_\infty$ in place of $w_\infty$. Then, by Theorem~\ref{thm: robust formulation}, the control policy
\begin{equation*}
    \mathbf{u} = \mathbf{K}_r\mathbf{y}+\bm{\eta}+\mathbf{w}\,,
\end{equation*}
can be applied to the real system during the exploration phase to generate safe trajectories. 

\section{Suboptimality Analysis}
\label{sec:suboptimality}
In this section, we tackle question Q2) in Section \ref{sec:prob_statement} about performance degradation as a function of the level of model-mismatch due to noisy data. We denote as $\starbf{K},\starbm{\Phi}$  the optimal controller for the real constrained problem \eqref{prob:IOP} and corresponding closed-loop responses.  Similarly, we denote as $\hatbf{K}^\star,\hatbm{\Phi}^\star$ the optimal controller for the optimization problem \eqref{prob:quasi_convex_alpha} and corresponding closed-loop responses. Further, we let $J^\star = J(\mathbf{G},\mathbf{K}^\star)$ and $\hat{J} = J(\mathbf{G},\hatbf{K}^\star)$. We aim to characterize the relative suboptimality gap  $\frac{\hat{J}^2-{J^\star}^2}{{J^\star}^2}$, and specifically we will show that
\begin{equation*}
    \frac{\hat{J}^2-{J^\star}^2}{{J^\star}^2} \leq \mathcal{O}\left(\epsilon_2\right) + \tilde{S}(\epsilon_\infty,\epsilon_2)\,,
\end{equation*}
 where $\tilde{S}(\epsilon_\infty,\epsilon_2) = S(\epsilon_\infty)(1+\mathcal{O}(\epsilon_2))$. Here, $S(\epsilon_\infty)$ quantifies the suboptimality incurred by tightening the constraints and is such that $S(0) = 0$. We prove that if $\epsilon_2$ and $\epsilon_\infty$ are small enough and the optimal controller $\mathbf{K}^\star$ does not activate the safety constraints, then $S(\epsilon_\infty) = 0$ and the suboptimality shrinks to $0$ linearly fast as $\epsilon_2$ converges to $0$. Otherwise, the gap may decrease according to $S(\epsilon_\infty)$, for which we provide a numerical plot in Section~\ref{sec:simulations}. In other words, for small estimation errors $\epsilon_2$ and $\epsilon_\infty$, applying controller $\hatbf{K}^\star$ (which is solely computed from noisy data) to the \emph{real} plant achieves almost optimal closed-loop performance while guaranteeing compliance with safety constraints. Surprisingly, despite the additional complexity of output-feedback and output noise, our bound matches the scaling with respect to $\epsilon = \max(\epsilon_2, \epsilon_\infty)$ that has been derived in \cite{dean2019safely} for the  state-feedback case without measurement noise. 

To prove the above statements, we first characterize a feasible solution to problem \eqref{prob:quasi_convex_alpha}, which we later exploit to establish our suboptimality bound. The proof of Lemma~\ref{le:feasible} and Theorem~\ref{th:suboptimality} is reported in the Appendices \ref{app:le:feasible} and \ref{app:th:suboptimality}, respectively.
\begin{lemma}[Feasible solution]
\label{le:feasible}
Let $\eta = \epsilon_2 \normtwo{\starbm{\Phi}_{uy}}$ and $\zeta = \epsilon_\infty \norm{\starbm{\Phi}_{uy}}_\infty $. Assume that the estimation errors are small enough to guarantee $\eta<\frac{1}{5}$ and $\zeta <\frac{1}{2}$,  and select $\alpha \in [\sqrt{2}\frac{\eta}{\epsilon_2(1-\eta)},\epsilon_2^{-1})$. 
Consider the following optimization problem and its optimal solutions $\bm{\Phi}^c$:
\begin{alignat}{3}
 \bm{\Phi}^c \in  & \argmin_{\bm{\Phi}} && \norm{\begin{bmatrix}
             \bm{\Phi}_{yy} & \bm{\Phi}_{yu} & \bm{\Phi}_{yy}\mathbf{y}_0\\\bm{\Phi}_{uy} & \bm{\Phi}_{uu} & \bm{\Phi}_{uy}\mathbf{y}_0
             \end{bmatrix}}_F \label{prob:doubly_robust}\\
&\st&&\begin{bmatrix}
             I&-\mathbf{G}
             \end{bmatrix}\bm{\Phi}=\begin{bmatrix}
             I&0
             \end{bmatrix}, \quad
             \bm{\Phi}
             \begin{bmatrix}
             -\mathbf{G}\\I
             \end{bmatrix}=\begin{bmatrix}
             0\\I
             \end{bmatrix},  \nonumber\\
        & && \norm{\bm{\Phi}_{uy}}_2 \leq \norm{\starbm{\Phi}_{uy}}_2\,, \norm{\bm{\Phi}_{uy}}_\infty \leq \norm{\starbm{\Phi}_{uy}}_\infty\,, \nonumber\\
        & && \phi_{1,j}(\bm{\Phi})+\phi_{2,j}(\bm{\Phi})+\phi_{3,j}(\bm{\Phi}) \leq \mathbf{b}_{y,j}\,, \label{eq:doublyrobust_constraints1}\\
     & && \phi_{4,j}(\bm{\Phi})+\phi_{5,j}(\bm{\Phi})+\phi_{6,j}(\bm{\Phi}) \leq \mathbf{b}_{u,j}\,,  \label{eq:doublyrobust_constraints2}\\
     & &&\forall j=1,\dots,sN\nonumber\,,\\
             & &&    \bm{\Phi}_{yy}, \bm{\Phi}_{yu}, \bm{\Phi}_{uy}, \bm{\Phi}_{uu} \text{\emph{ with causal sparsities}} \nonumber \,.
\end{alignat} 
where
\begin{align*}
    &\phi_{1,j}(\bm{\Phi}) = \frac{v_\infty \norm{F_{y,j} \bm{\Phi}_{yy}}^{\star}_1}{1-2\zeta}\,, \quad \phi_{4,j}(\bm{\Phi}) = \frac{v_\infty \norm{F_{u,j} \bm{\Phi}_{uy}}^{\star}_1}{1-2\zeta}\,,\\
    & \phi_{2,j}(\bm{\Phi}) = w_\infty\norm{\begin{bmatrix}\left(F_{y,j}\bm{\Phi}_{yu}\right)^\mathsf{T}\\2\frac{\epsilon_\infty+\zeta\norm{\hatbf{G}}_\infty}{1-2\zeta}\left(F_{y,j}\bm{\Phi}_{yy}\right)^\mathsf{T}\end{bmatrix}}_1\,,\\
    &\phi_{5,j}(\bm{\Phi}) = w_\infty\norm{\begin{bmatrix}\left(F_{u,j}\bm{\Phi}_{uu}\right)^\mathsf{T}\\2\frac{\epsilon_\infty+\zeta\norm{\hatbf{G}}_\infty}{1-2\zeta}\left(F_{u,j}\bm{\Phi}_{uy}\right)^\mathsf{T}\end{bmatrix}}_1\,,\\
    &\phi_{3,j}(\bm{\Phi}) = F_{y,j} \bm{\Phi}_{yy}\hatbf{y}_0 + 2\frac{\epsilon_\infty+\zeta\norm{\hatbf{y}_0}_\infty}{1-2\zeta}\norm{F_{y,j}\bm{\Phi}_{yy}}^{\star}_1\,,\\
    &\phi_{6,j}(\bm{\Phi}) = F_{u,j} \bm{\Phi}_{uy}\hatbf{y}_0 + 2\frac{\epsilon_\infty+\zeta\norm{\hatbf{y}_0}_\infty}{1-2\zeta}\norm{F_{u,j}\bm{\Phi}_{uy}}^{\star}_1\,.
\end{align*}

Then, the following expressions
    \begin{align}
        \tildebm{\Phi}_{yy} &= \bm{\Phi}^c_{yy} (I \hspace{-0.08cm}+ \hspace{-0.08cm}\mathbf{\Delta} \bm{\Phi}^c_{uy})^{-1},~ 
        \tildebm{\Phi}_{yu} = \bm{\Phi}^c_{yy} (I \hspace{-0.08cm}+\hspace{-0.08cm} \mathbf{\Delta} \bm{\Phi}^c_{uy})^{-1}(\mathbf{G} \hspace{-0.08cm}-\hspace{-0.08cm} \mathbf{\Delta}), \nonumber\\
        \tildebm{\Phi}_{uy} &= \bm{\Phi}^c_{uy} (I + \mathbf{\Delta} \bm{\Phi}^c_{uy})^{-1},~
        \tildebm{\Phi}_{uu} = (I +  \bm{\Phi}^c_{uy}\mathbf{\Delta})^{-1}\bm{\Phi}^c_{uu} , \nonumber \\
        \widetilde{\gamma} &= \frac{\sqrt{2}\eta}{\epsilon_2(1-\eta)}, \qquad \widetilde{\tau} = \frac{\zeta}{\epsilon_\infty(1-\zeta)}\,,   \label{eq:suboptimal}
    \end{align}
provide a feasible solution to problem \eqref{prob:quasi_convex_alpha}.
\end{lemma}

\vspace{0.2cm}
The main idea behind Lemma~\ref{le:feasible} is to construct a feasible solution to problem \eqref{prob:quasi_convex_alpha} from the set of closed-loop responses generated applying a cautious ground-truth optimal controller $\mathbf{K}^c = \bm{\Phi}_{uy}^c(\bm{\Phi}_{yy}^c)^{-1}$ on the estimated system $\hatbf{G}$. In the absence of safety constraints, such a feasible solution could directly be established from the ground-truth optimal policy $\starbf{K}$ similar to \cite{zheng2020sample}. In the constrained case, however, one cannot expect the optimal solution $\starbf{K}$ to be feasible for $\hatbf{G}$ in \eqref{prob:quasi_convex_alpha} since \eqref{eq:quasiconvex_constraints1}-\eqref{eq:quasiconvex_constraints2} are more stringent than 
\eqref{eq:ach_constraints1}-\eqref{eq:ach_constraints2}. Hence, in \eqref{prob:doubly_robust} we first compute $\mathbf{K}^c = \bm{\Phi}_{uy}^c(\bm{\Phi}_{yy}^c)^{-1}$ as the optimal linear policy for the real system $\mathbf{G}$ under safety constraints that are more stringent than those of \eqref{prob:quasi_convex_alpha}, and subsequently define $\tildebm{\Phi}$ as the closed-loop responses generated applying $\mathbf{K}^c$ to $\hatbf{G}$. In this way, $\tildebm{\Phi}$ is guaranteed to be feasible for \eqref{prob:quasi_convex_alpha}, provided that the model mismatch is sufficiently small.

Clearly, the optimal solution $\mathbf{K}^c = \bm{\Phi}_{uy}^c(\bm{\Phi}_{yy}^c)^{-1}$ to \eqref{prob:doubly_robust} will yield a suboptimal cost $J(\mathbf{G},\mathbf{K}^c)\geq J(\mathbf{G},\mathbf{K}^\star)$.  We denote the corresponding suboptimality gap as
\begin{equation}
\label{eq:sub_gap_S}
    S(\epsilon_\infty) = \frac{J(\mathbf{G},\mathbf{K}^c)^2-J(\mathbf{G,\mathbf{K}^\star})^2}{J(\mathbf{G,\mathbf{K}^\star})^2}\,.
\end{equation}
Note that, if  the estimation error $\epsilon_\infty$ is too large, the optimization problem \eqref{prob:doubly_robust} may become infeasible. This is expected as the uncertainty level might be incompatible with the required safety.  On the other hand, if the optimal solution to the non-noisy problem \eqref{prob:IOP} does not activate the safety constraints, then the constraints of \eqref{prob:doubly_robust} remain inactive for small enough $\epsilon_\infty$. In such case we have that $S(\epsilon_\infty) = 0$.

We are now ready to state the main suboptimality result.
\begin{theorem}
\label{th:suboptimality}
Let $\eta = \epsilon_2 \normtwo{\starbm{\Phi}_{uy}}$ and $\zeta = \epsilon_\infty \norm{\starbm{\Phi}_{uy}}_\infty $. Assume that the estimation errors are small enough to guarantee $\eta<\frac{1}{5}$ and $\zeta <\frac{1}{2}$,  and select $\alpha \in \left[\sqrt{2}\frac{\eta}{\epsilon_2(1-\eta)},5\norm{\starbm{\Phi}_{uy}}_2\right]$. Moreover, assume that $\epsilon_\infty $ is small enough for the optimization problem \eqref{prob:doubly_robust} to be feasible. Then, when applying the controller $\hatbf{K}^\star$ optimizing \eqref{prob:quasi_convex} to the true plant $\mathbf{G}$, the relative error with respect to the true optimal cost is upper bounded as
\begin{align}
        \frac{\hat{J}^2-{J^\star}^2}{{J^\star}^2} &\leq  20\eta  +4 (M^c+V^c)+ 4S(\epsilon_\infty)(1+M^c+V^c)\nonumber\\
       &= \mathcal{O}\left(\epsilon_2 \left(1+\norm{\starbm{\Phi}_{uy}}_2\right)\left(1+\norm{\mathbf{G}}_2+\norm{\mathbf{y}_0}_2\right)^2\right) +\nonumber\\
       &~~+4S(\epsilon_\infty)(1+M^c+V^c)\,,\label{eq:sub_gap_th}
\end{align}
where 
\begin{align*}
   & M^c = h(\epsilon_2,\alpha,\hatbf{G})+ h(\epsilon_2,\alpha,\hatbf{y}_0)+ h(\epsilon_2,\norm{\cbm{\Phi}_{uy}}_2,\mathbf{G})\\
   &\qquad + h(\epsilon_2,\norm{\cbm{\Phi}_{uy}}_2,\mathbf{y}_0)\,,\\
      & V^c = h(\epsilon_2,\alpha,\hatbf{y}_0)+h(\epsilon_2,\norm{\cbm{\Phi}_{uy}}_2,\mathbf{y}_0)\,.
\end{align*}
\end{theorem}

\vspace{0.2cm}

We have expressed the suboptimality gap in the form \eqref{eq:sub_gap_th} to highlight the presence of two main parts; the first addend scales as $\mathcal{O}\left(\epsilon_2 \left(1+\norm{\starbm{\Phi}_{uy}}_2\right)\left(1+\norm{\mathbf{G}}_2+\norm{\mathbf{y}_0}_2\right)^2\right)$ and the second addend $S(\epsilon_\infty)(1+M^c+V^c)$ is linked to the suboptimality of the tightened optimization program \eqref{prob:doubly_robust}.   The most important observation is that the suboptimality  decreases at most \emph{linearly} with $\epsilon_2$ when $\max{(\epsilon_2,\epsilon_\infty)}$ is small enough. A linear suboptimality rate in the output-feedback case has first been observed for the unconstrained setup of \cite{zheng2020sample}. Recovering a similar suboptimality rate for the general case with constraints is one of the main novelties of our work. Indeed, despite recovering an upper bound that scales similarly to \cite{zheng2020sample}, the corresponding analysis in Appendix~\ref{app:th:suboptimality} is significantly complicated by the fact that the feasible solution used in \cite{zheng2020sample} cannot be exploited anymore. Hence, one might expect that the suboptimality rate will worsen with respect to the unconstrained case of \cite{zheng2020sample}. Theorem~\ref{th:suboptimality} shows, however, that the bound does not deteriorate for small-enough model mismatch levels. Turning our attention to the term $S(\epsilon_\infty)$,  we observe through examples (cfr. Figure~\ref{fig:S_plot}) that $S(\epsilon_\infty)$ sharply transitions from $0$ to $\infty$ as $\epsilon_\infty$ increases. In practice, this example suggests that $S(\epsilon_\infty)$ might be interpreted as an indicator function; if $S(\epsilon_\infty)\approx 0$, then $\epsilon_\infty$ is small enough for the linear  suboptimality rate to hold.

Our suboptimality bound \eqref{eq:sub_gap_th} indicates features of the underlying unknown system that make it easier to be safely controlled based on noisy data. Notably, the suboptimality grows quadratically with the norm of the true impulse and free responses. This fact implies that an unknown unstable system will be more difficult to control for a long horizon. Last, we note that, surprisingly, our rate in terms of $\epsilon_2$ matches that of \cite{dean2019safely} which was valid under the assumption of exact state measurements. In other words, our analysis shows that near-optimality can be ensured in complex data-driven control scenarios that combine hard safety requirements with noisy output measurements. 

\section{Numerical Experiments}
\label{sec:simulations}
\begin{figure*}[t!]
    \centering
    \begin{subfigure}[t]{0.45\textwidth}
        \includegraphics[width = 1\linewidth]{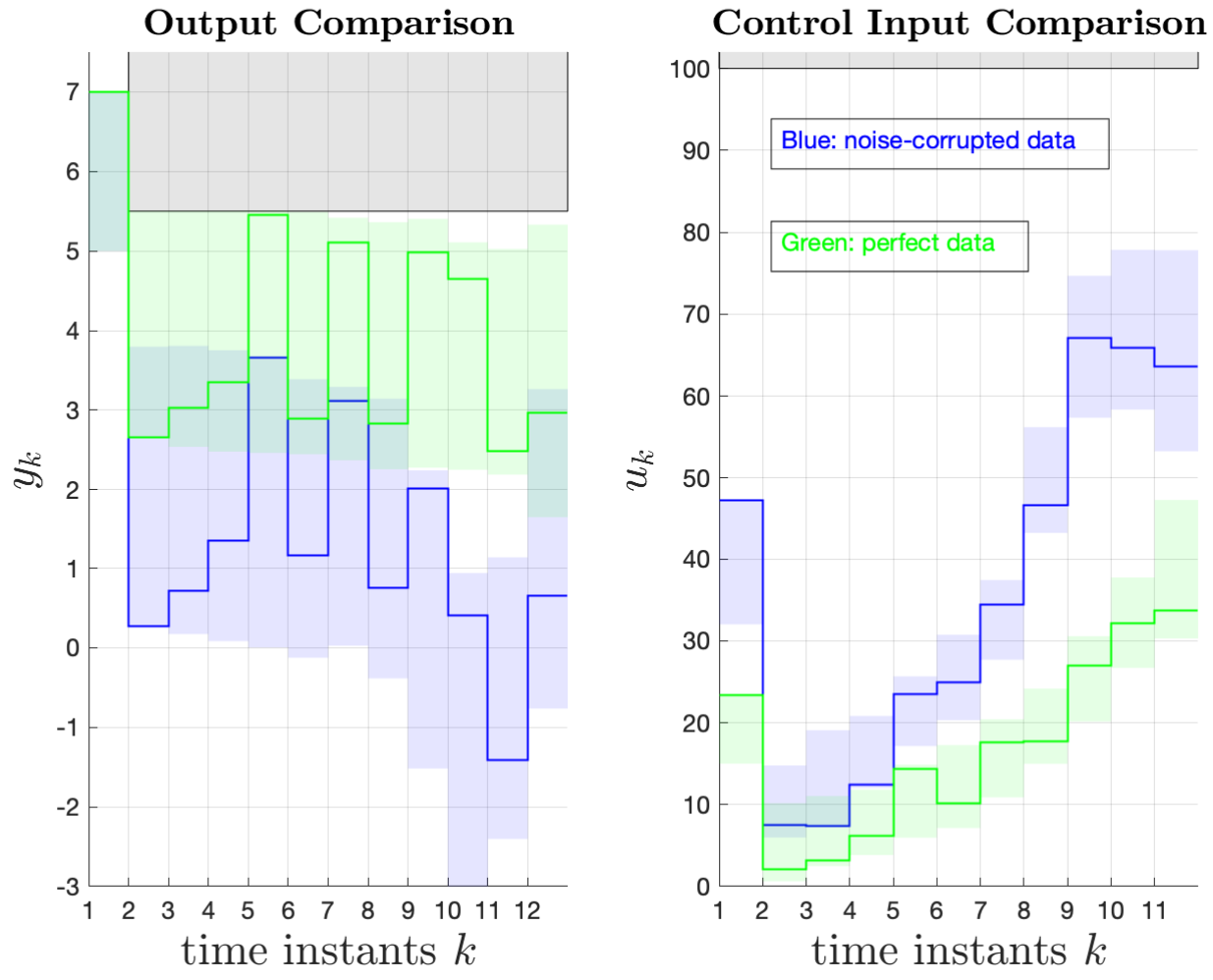}
        \caption{Closed-loop trajectories. The grey region indicates unsafe input and output values. The green and blue regions contain trajectories for $50$ noise realizations obtained through $\mathbf{K}^\star$ and $\hatbf{K}^\star$, respectively. Green and blue lines represent a specific trajectory in both settings.}
        \label{fig:constrained_traj}
    \end{subfigure}%
    ~ 
    \begin{subfigure}[t]{0.45\textwidth}
    \centering
        \includegraphics[width = .595\linewidth]{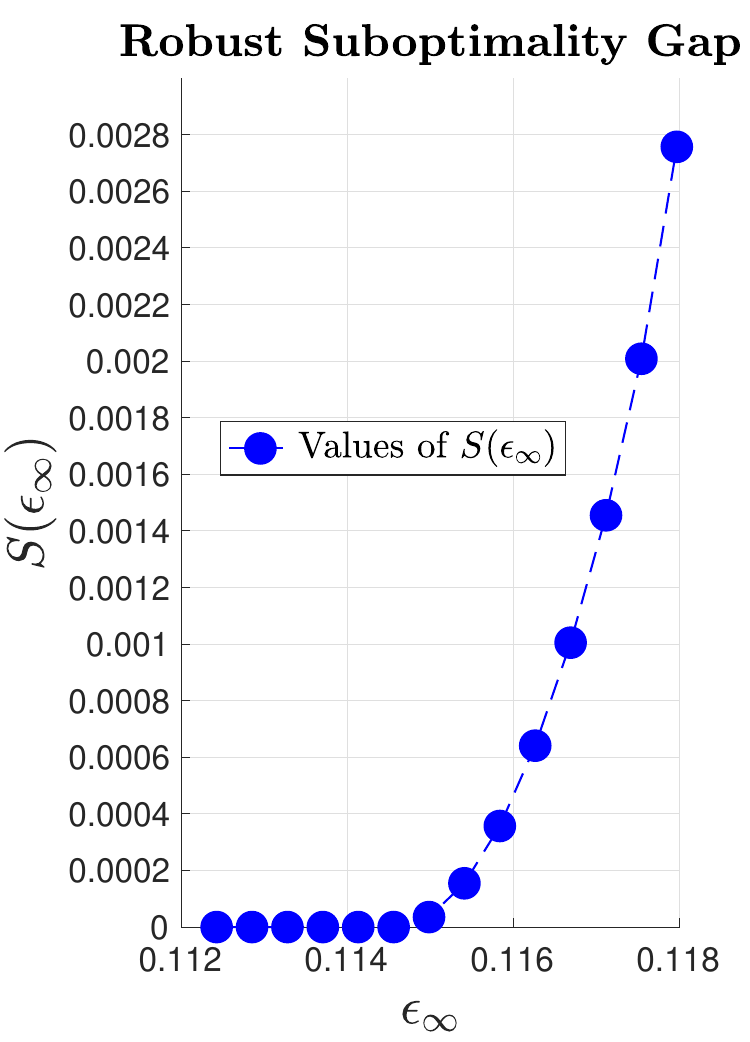}
        \caption{Robust suboptimality gap $S(\epsilon_\infty)$. This quantity can be interpreted as an indicator as to whether the guarantee \eqref{eq:sub_gap_th} holds for a given $\epsilon_\infty$.}
        \label{fig:S_plot}
    \end{subfigure}
    \caption{Safe controller synthesis for system \eqref{eq:system_example}.}
\end{figure*}
In this section, we demonstrate numerically the effectiveness of the proposed framework in safely controlling unknown systems. In the experiments, we consider the single-input single-output unknown LTI system characterized by the matrices
\begin{align}
&A = \rho \begin{bmatrix}
       1 &  0.25\\0 & 1
\end{bmatrix}, \text{ } B =  \begin{bmatrix}
      0 \\  0.1
\end{bmatrix},\text{ } C =  \begin{bmatrix}
       1 & -1
\end{bmatrix}\,, \label{eq:system_example}
\end{align}
where $\rho>0$ corresponds to the spectral radius of $A$. When $\rho<1$, \eqref{eq:system_example} is asymptotically stable, that is, its output converges to the origin at an exponential rate when the input is equal to $0$. When $\rho = 1$, \eqref{eq:system_example} is a marginally stable double-integrator system. 

In all the following tests, the cost function is given by \eqref{eq:cost_output} for appropriate choices of the weights. The expectation in \eqref{eq:cost_output} is taken over future input/output disturbances with  covariance matrices $\Sigma_w = I_m$ and $\Sigma_v = I_p $. We consider bounded disturbances between $-1$ and $1$, that is, $w_\infty = v_\infty = 1$. Hence, each scalar disturbance is randomly chosen from $\{-1,1\}$ with probability $\frac{1}{2}$.   For solving optimization problems we use MOSEK \cite{mosek}, called through MATLAB via YALMIP \cite{YALMIP}\footnote{The code is open-source and available at \url{https://gitlab.nccr-automation.ch/data-driven-control-epfl/constrained-biop}. This example takes a few minutes overall to run.}.

\subsection{Example: safe controller synthesis from noisy data}
In our first test, we synthesize  a safe output-feedback controller for system \eqref{eq:system_example} with $\rho = 1$ from noisy data. We assume that  $x_0 = x(1) = \begin{bmatrix}6&0\end{bmatrix}^\mathsf{T}$, where we set the initial time at $t=1$ rather than $t=0$ for compliance with MATLAB's indexing of vector entries.

The safety constraints are: $y(1) \in \mathbb{R}$ and
\begin{align*}
    -5.5&\leq y(t) \leq 5.5,\quad \forall t = 2,\dots,12\,, \\
    -100&\leq u(t)\leq 100\,, \quad \forall t = 1,\dots,11\,,
\end{align*}
for all realizations of noise $\|\mathbf{w}\|_\infty,\|\mathbf{v}\|_\infty\leq 1$, while minimizing the cost \eqref{eq:cost_output} with the weights $\mathbf{Q}$ and $\mathbf{R}$ in \eqref{prob:IOP} set to the identity. 

We first synthesize the optimal controller assuming that the available data are not affected by noise. To this end, we cast and solve the convex optimization problem \eqref{prob:IOP_data}. We verify that the optimal controller $\mathbf{K}^\star$ yields a cost $J(\mathbf{G},\mathbf{K}^\star) = 69.88$. 
The green tubes in Figure~\ref{fig:constrained_traj} show the regions containing 50 realizations of the optimal closed-loop input and output trajectories. Due to the high level of noise, we can observe a significant variability in the trajectory values for different noise realizations. Nonetheless, all trajectories are safe.

We then discuss the case where the available data are affected by noise. In order for the tightened constraints of \eqref{prob:quasi_convex_alpha} to be feasible, we consider noisy estimates $(\hatbf{G},\hatbf{y}_0)$ with $\epsilon = 0.01$ and compute a near-optimal solution to the proposed optimization problem \eqref{prob:quasi_convex}. As discussed in Section~\ref{sub:proposed_reformulation}, this can be achieved by 1) extensive or random search over $\gamma$ and $\tau$, or 2) extensive search over $\tau$ and golden-section search over $\gamma$. Even if the first solution comes without strong theoretical guarantees, extensive search over $\gamma$ and $\tau$ may be simpler to implement as it avoids the delicate task of tuning the parameter $\alpha$. Specifically, for this example we have searched over $100$ randomly extracted values of $\gamma$ and $\tau$ in the interval $[0,\epsilon^{-1})$. A potential improvement to this heuristic could be to use a bisection algorithm, as proposed in \cite{chen2020robust} for example.

Proceeding as above, we synthesize a robustly safe controller $\hatbf{K}^\star$ yielding a cost of $J(\mathbf{G},\hatbf{K}^\star) =   140.54$. The corresponding suboptimality gap is $\frac{\hat{J}^2-{J^\star}^2}{{J^\star}^2} = 3.049$. In Figure~\ref{fig:constrained_traj}, the trajectories and variability levels resulting from $\hatbf{K}^\star$ for $50$ noise realizations are plotted in blue. We observe that, since $\hatbf{K}^\star$ is synthesized using noise-corrupted data, it leads to safer, but more conservative trajectories. Indeed, due to uncertainty, higher control effort is spent to keep the output further from the constraints. 

It is informative to inspect the robust suboptimality gap $S(\epsilon_\infty)$ incurred by the tightened optimization problem \eqref{prob:doubly_robust} that we have used in the analysis to characterize a feasible solution to \eqref{prob:quasi_convex_alpha}. In Figure~\ref{fig:S_plot}, we plot $S(\epsilon_\infty)$ assuming $x_0 = x(1) = \begin{bmatrix}1&0\end{bmatrix}^\mathsf{T}$ and requiring $-3\leq y(t)\leq 3$ for $t = 1,\dots,7$. The example exhibits a fast transition from infeasibility for $\epsilon_\infty>0.118$ to near-optimality for $\epsilon_\infty<0.115$.  This fact leads to the following observation: high-performing safe controllers can be synthesized by solving  \eqref{prob:quasi_convex_alpha} even when the optimization problem \eqref{prob:doubly_robust} is infeasible, i.e. $S(\epsilon_\infty) = \infty$. In such cases the suboptimality bound \eqref{eq:sub_gap_th} is not applicable, but a robustly safe controller has been synthesized nonetheless. This phenomenon is consistent with the numerical examples of \cite{dean2019safely} for the state-feedback case. 

\begin{figure*}[t!]
    \centering
    \begin{subfigure}[t]{0.468\textwidth}
        \includegraphics[width = 1.02\linewidth]{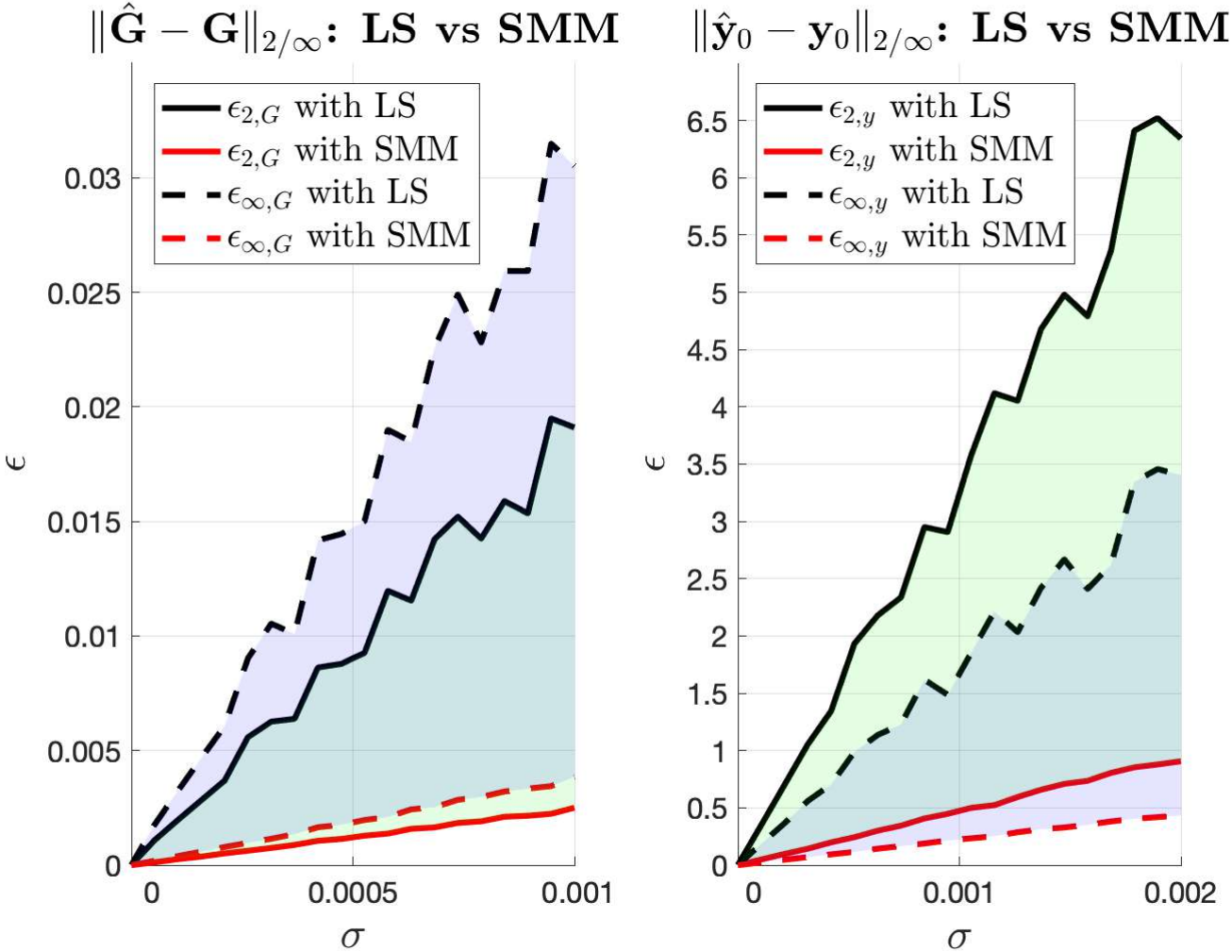}
        \caption{Estimation error in function of the corrupting noise. ML estimation through the SMM yields significantly smaller errors than LS. The green and blue regions indicate the gap for the $2$-norm and $\infty$-norm, respectively.}
        \label{fig:SMM_gap}
    \end{subfigure}%
    ~ 
    \begin{subfigure}[t]{0.45\textwidth}
    \centering
        \includegraphics[width = 1.02\linewidth]{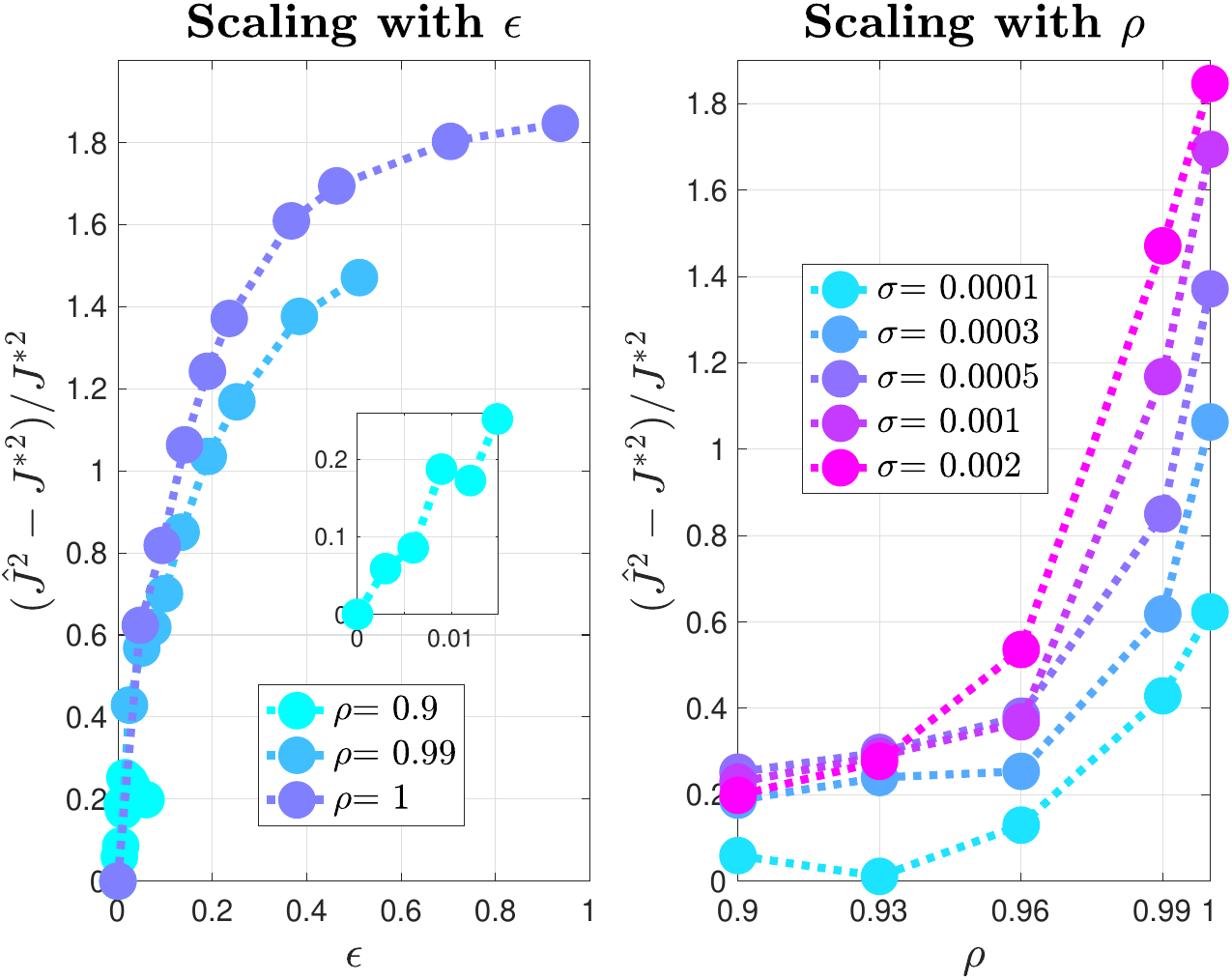}
        \caption{ Suboptimality gap as a function of $\epsilon_2$ (obtained through SMM estimation) for increasing values of the spectral radius $\rho$ of matrix $A$ (on the left). Suboptimality gap as a function of $\rho$ for increasing values of $\sigma$ (on the right).}
        \label{fig:Suboptimality}
    \end{subfigure}
    \caption{Examples for estimation and suboptimality scaling.}
\end{figure*}

\subsection{Example: suboptimality scaling beyond least-squares estimation}

The bound \eqref{eq:sub_gap_th} in Theorem~\ref{th:suboptimality} states that a low estimation error level $\epsilon$ is crucial in ensuring safety and near-optimality when controlling unknown systems based on noisy data. One advantage of the proposed formulation is that it is directly compatible with behavioral estimation approaches beyond LS identification for the reconstruction of the impulse and free responses, such as data-enabled Kalman filtering \cite{alpago2020extended} and SMM  \cite{yin2021maximum,iannelli2020experiment}. In our last test, we  drop the constraints for both the input and the outputs thus putting our focus on 1) validating the linear scaling of the suboptimality gap \eqref{eq:sub_gap_th}, and 2) showcasing  that, for instance, SMM-based estimation \cite{iannelli2020experiment} may lead to significantly lower error levels given the same amount of data. We consider the system \eqref{eq:system_example} with different values of $\rho\in [0.9,0.93,0.96,0.99,1]$ and $x_0 = x(1) = \begin{bmatrix}6&0\end{bmatrix}^\mathsf{T}$, over a time-horizon of length $N= 11$.  The cost function weights in \eqref{prob:IOP} are selected as $Q(t) = I_p$ for every $t = 1,\dots,11$, $Q(12) = 20I_p$ and $R(t)=0.05$ for every $t = 1,\dots, 11$. 
 
\subsubsection{Behavioral estimation: LS vs SMM }
For a fixed value of $\rho$, we gather system trajectories of length $200$ time-steps which are corrupted by input and output Gaussian noise with covariance matrices equal to $\sigma I$. For each experiment, we fix the variance $\sigma \geq 0$ and select a random exploration control input $\mathbf{u}$. We collect $1000$ different trajectories for different realizations of the corrupting noise. For each realization of the trajectories, we compute 1) the LS solution $(G_{LS},g_{LS})$ using \eqref{eq:G_LS} and the corresponding impulse and free responses $\widetilde{\mathbf{G}}_{LS},\widetilde{\mathbf{y}}_{0,LS}$, and 2) the ML solution $(G_{ML},g_{ML})$ using  \eqref{eq:G_ML}-\eqref{eq:g_ML} and the corresponding impulse and free responses $\widetilde{\mathbf{G}}_{ML},\widetilde{\mathbf{y}}_{0,ML}$. For each estimation, we determine the incurred error levels $\epsilon_{2,G}$, $\epsilon_{\infty,G}$, $\epsilon_{2,y}$ and $\epsilon_{\infty,y}$\footnote{Since the real system is unavailable, in practice this can be done using a bootstrap procedure.}. Last, we record the $90$-th percentile of these values, both for SMM and LS estimation.
 
  In Figure~\ref{fig:SMM_gap} we compare the values of $\epsilon_2$ and $\epsilon_\infty$ incurred by both estimation techniques. We observe that SMM may yield significantly smaller estimation errors than LS identification. While a full sample-complexity analysis is still unavailable beyond least-squares \cite{dean2019sample,oymak2019non,xue2021data}, these examples showcase an advantage in using more sophisticated estimation techniques for safe data-driven control.

\subsubsection{Suboptimality scaling}
Having exploited ML estimation to construct approximate impulse and free responses and the corresponding error-levels, we are ready to solve the optimization problem \eqref{prob:quasi_convex_alpha}. Since constraints are not present in this example, \eqref{prob:quasi_convex_alpha} can be simplified to the quasiconvex formulation we have proposed in \cite{furieri2021behavioral}, where the optimization variable $\tau$ is not present. The parameter $\alpha$ is tuned empirically  in the interval $\alpha \in [\sqrt{2}\frac{\eta}{\epsilon_2(1-\eta)},\epsilon_2^{-1})$.\footnote{The value $\eta=\epsilon_2 \norm{\starbm{\Phi}_{uy}}_2$ is unknown in practice because $\starbm{\Phi}_{uy}$ is unavailable. One can then tune $\alpha$ according to $\alpha < \epsilon_2^{-1}$.}

Figure~\ref{fig:Suboptimality} shows the suboptimality gap one incurs by applying the controller $\hatbf{K}^\star$ obtained through the proposed approach. On the left,  we consider increasing levels of the estimation error level $\epsilon_2$ for each choice of the spectral radius $\rho= 0.9,0.99,1$. On the right, we conversely consider increasing levels of the spectral radius for each choice of the estimation error level $\epsilon_2$. In both cases, we plot the suboptimality gap $\frac{\hat{J}^2-{J^\star}^2}{{J^\star}^2}$. It can be observed that, consistently with Theorem~\ref{th:suboptimality}, 1) the gap linearly converges to $0$ as $\epsilon_2$ converges to $0$, and 2) the gap may grow faster than linearly with the spectral radius $\rho$ as a larger $\rho$ generally leads to larger $\norm{\mathbf{G}}_2$. We also observe that larger $\rho$ may lead to higher model mismatch values $\epsilon$. Finally, we remark that, in finite-horizon, our formulations are valid for unstable systems with $\rho>1$. However,  it is inherently challenging to collect trajectories  of an unstable system, as the values to be plugged into the corresponding optimization problems will become too large to be handled by numerical solvers.  For unstable systems in a data-driven scenario, it is common to assume knowledge of a pre-stabilizing controller \cite{simchowitz2020improper,zheng2020sample}.

\section{Conclusions}
\label{sec:conclusions}
In this paper, we have analyzed how much the model-mismatch due to noisy data can impact the safety and performance of output-feedback control systems with constraints. By deriving a suitable problem relaxation, we have proven that, despite the presence of constraints, the suboptimality of our proposed problem relaxation increases at most linearly for small model mismatches incurred during system identification.

While the proposed approach can synthesize safe and near-optimal output-feedback controllers from noisy data, our relaxed problem might be infeasible for fairly small error levels. Feasibility issues may be significantly mitigated by using soft and chance constraints, or ellipsoidal model mismatch sets such as those of \cite{abbasi2011regret, yin2021data}. Future work also includes analyzing the suboptimality in a receding-horizon setup using closed-loop predictions, as well as investigating the advantages of directly optimizing based on the data trajectories rather than performing an identification step.

\section*{Acknowledgments}
We thank Sarah Dean for sharing the implementation of the examples in \cite{dean2019safely} and for helpful insights on quasiconvexity. We also thank the anonymous reviewers for the  suggestions for improvement, as well as several new interesting insights.

\bibliographystyle{IEEEtran}
\bibliography{references}

\appendix
\allowdisplaybreaks

\subsection{Willems' lemma and behavioral theory for synthesizing safe controllers}
\label{app:Willems}

We recall the definition of persistency of excitation and the result known as the \emph{Fundamental Lemma} for LTI systems \cite{willems2005note}.
\begin{definition}
We say that $\mathbf{u}^h_{[0,T-1]}$ is \emph{persistently exciting} (PE) of order $L$ if the Hankel matrix $\mathcal{H}_L(\mathbf{u}^h_{[0,T-1]})$ is full row-rank.
\end{definition}
A necessary condition for the matrix $\mathcal{H}_{L}(\mathbf{u}^h_{[0,T-1]})$ to be full row-rank is that it has at least as many columns as rows. It follows that the input trajectory $\mathbf{u}^h_{[0,T-1]}$  must be long enough to satisfy $T\geq (m+1)L-1$.
\begin{lemma}[Theorem 3.7, \cite{willems2005note}]
\label{le:Willems}
Consider system \eqref{eq:dynamic}. Assume that $(A,B)$ is controllable and that there is no noise. Let $\{\mathbf{y}^h_{[0,T-1]},\mathbf{u}^h_{[0,T-1]}\}$ be a system trajectory of length $T$ that has been recorded during a past experiment. Then, if $\mathbf{u}^h_{[0,T-1]}$ is PE of order $n+L$, the signals $\mathbf{y}^\star_{[0,L-1]}\in \mathbb{R}^{pL}$ and $\mathbf{u}^\star_{[0,L-1]} \in \mathbb{R}^{mL}$ are trajectories of \eqref{eq:dynamic} if and only if there exists $g \in \mathbb{R}^{T-L+1}$ such that
\begin{equation}
    \begin{bmatrix}
           \mathcal{H}_{L}(\mathbf{y}^h_{[0,T-1]})\\
           \mathcal{H}_{L}(\mathbf{u}^h_{[0,T-1]})
    \end{bmatrix} g = \begin{bmatrix}\mathbf{y}^\star_{[0,L-1]}\\\mathbf{u}^\star_{[0,L-1]}\end{bmatrix}\,. \label{eq:Willems}
\end{equation}
\end{lemma}

\vspace{0.2cm}

We proceed by showing  how Lemma~\ref{le:Willems} allows one to derive a data-driven formulation of \eqref{prob:IOP} when the data are not noisy. We work under the following assumptions that are standard in the behavioral framework.

\begin{assumption}
\label{ass:1}
The data-generating LTI system \eqref{eq:dynamic} is such that $(A,B)$ is controllable and $(A,C)$ is observable.
\end{assumption}

\begin{assumption}
\label{ass:4}
The historical input trajectory $\mathbf{u}^h_{[0,\tilde{T}-1]}$ is PE of order $n+T_{ini}+N$, where $T_{ini}\geq l$ and $l$ is the smallest integer such that the matrix \begin{equation*}\begin{bmatrix}C^\mathsf{T}&(CA)^\mathsf{T}&\dots &(CA^{l-1})^\mathsf{T}\end{bmatrix}^\mathsf{T}\,,\end{equation*} has full row-rank. Note that if Assumption~\ref{ass:1} holds, then $l \leq n$.
\end{assumption}

Further, we give the following definition.
\begin{definition}
\label{ass:2}
The available data in $\mathbf{D1}$ are further split as follows:
    \begin{enumerate}
    \item[$i)$] a \emph{recent} system trajectory of length $T_{ini}$: $\left\{\mathbf{y}^r_{[0,T_{ini}-1]},\mathbf{u}^r_{[0,T_{ini}-1]}\right\}$, with $\mathbf{y}^r_{[0,T_{ini}-1]} \hspace{-0.0175cm} = \hspace{-0.0175cm} \mathbf{y}_{[-T_{ini},-1]}$ and $\mathbf{u}^r_{[0,T_{ini}-1]} = \mathbf{u}_{[-T_{ini},-1]}$,
        \item[$ii)$] a \emph{historical} system trajectory of length $\tilde{T}$: $        \left\{\mathbf{y}^h_{[0,\tilde{T}-1]},\mathbf{u}^h_{[0,\tilde{T}-1]}\right\}$, with $\mathbf{y}^h_{[0,\tilde{T}-1]} = \mathbf{y}_{[-T_{h},-T_{h}+\tilde{T}-1]}$ and $\mathbf{u}^h_{[0,\tilde{T}-1]} = \mathbf{u}_{[-T_{h},-T_h+\tilde{T}-1]}$ for $T_h \in \mathbb{N}$ such that $T_h \geq \tilde{T}$ and $\tilde{T}\leq T$.
    \end{enumerate}
\end{definition}

\vspace{0.2cm}

The \emph{historical} data are to be used in substitution of the system model, while the \emph{recent} data  reflect the system initial state $x_0 \in \mathbb{R}^n $ \cite{markovsky2008data}. By exploiting \eqref{eq:Willems}, one can derive a constrained version of the BIOP derived in \cite{furieri2021behavioral} as follows:
\begin{proposition}[Safe Behavioral IOP]
\label{th:BIOP}
Consider the LTI system \eqref{eq:dynamic}, whose parameters $(A,B,C,x_0)$ are \emph{unknown}, and let Assumptions~\ref{ass:1}-\ref{ass:4} hold. Further assume that the historical and recent trajectories are \emph{not} affected by noise. Let $(G,g)$ be any solutions to the linear system of equations
\begin{equation}
\label{eq:BM}
    \begin{bmatrix}U_p\\ Y_p\\U_f\end{bmatrix}
    \hspace{-0.1cm}
    \begin{bmatrix} G&g \end{bmatrix}
    \hspace{-0.06cm}
    =
    \hspace{-0.06cm}
    \begin{bmatrix}0_{mT_{ini} \times m}&\mathbf{u}^r_{[0,T_{ini}-1]}\\0_{pT_{ini} \times m} & \mathbf{y}^r_{[0,T_{ini}-1]}\\\begin{bmatrix}I_m&0_{ m \times m(N-1)} \end{bmatrix}^\mathsf{T} &0_{mN \times 1}\end{bmatrix}\hspace{-0.1cm},
\end{equation}
where $\begin{bmatrix}U_p\\U_f\end{bmatrix} = \mathcal{H}_{T_{ini}+N}(\mathbf{u}^h_{[0,\tilde{T}-1]})$ and   $\begin{bmatrix}Y_p\\Y_f\end{bmatrix} = \mathcal{H}_{T_{ini}+N}(\mathbf{y}^h_{[0,\tilde{T}-1]})$. Then, the optimization problem~\eqref{prob:IOP} is equivalent to 
\begin{align}
    &\min_{\bm{\Phi}}\frobenius{\begin{bmatrix}\mathbf{Q}^{\frac{1}{2}}&0\\0&\mathbf{R}^{\frac{1}{2}}\end{bmatrix}
    \hspace{-0.05cm}
    \begin{bmatrix}
            \bm{\Phi}_{yy} & \bm{\Phi}_{yu} \\
            \bm{\Phi}_{uy} & \bm{\Phi}_{uu} 
        \end{bmatrix}
        \hspace{-0.05cm}
        \begin{bmatrix}\bm{\Sigma}^{\frac{1}{2}}_v&0&Y_fg\\0&\bm{\Sigma}^{\frac{1}{2}}_w&0\end{bmatrix}}^2\label{prob:IOP_data}\\
    &\st \begin{bmatrix} I & -\operatorname{Toep}(Y_fG) \end{bmatrix}\bm{\Phi} = \begin{bmatrix} I & 0 \end{bmatrix}\,,\nonumber \\
    &~~~~\qquad \quad\bm{\Phi}\begin{bmatrix} -\operatorname{Toep}(Y_fG) \\I \end{bmatrix} = \begin{bmatrix} 0 \\ I\end{bmatrix} \,,\nonumber\\
           &\qquad\quad \norm{\begin{bmatrix}v_\infty\left(F_{y,j} \bm{\Phi}_{yy}\right)^\mathsf{T}\\w_\infty\left(F_{y,j} \bm{\Phi}_{yu}\right)^\mathsf{T}\end{bmatrix}^{\mathsf{T}}}_1+  \left(F_{y,j} \bm{\Phi}_{yy}\right)Y_fg \leq \mathbf{b}_{y,j}\,, \nonumber\\
            &\qquad\quad\norm{\begin{bmatrix}v_\infty\left(F_{u,j}\bm{\Phi}_{uy}\right)^\mathsf{T}\\w_\infty\left(F_{u,j} \bm{\Phi}_{uu}\right)^\mathsf{T} \end{bmatrix}^{\mathsf{T}}}_1 + \left(F_{u,j} \bm{\Phi}_{uy}\right)Y_fg \leq \mathbf{b}_{u,j} \,,\nonumber \\
            &\qquad \quad \forall j=1,\dots,sN\,, \nonumber\\
                   &\qquad\quad\bm{\Phi}_{yy}, \bm{\Phi}_{yu}, 
            \bm{\Phi}_{uy}, \bm{\Phi}_{uu}  \text{ with causal sparsities.}\nonumber
\end{align}
\end{proposition}
The proof of Proposition~\ref{th:BIOP} is analogous to that of Theorem~1 in \cite{furieri2021behavioral}, with the addition of the safety constraints as per Proposition~\ref{prop:IOP}. Since the historical and recent data are not noisy, $Y_fG$ and $Y_fg$ yield the \emph{true} impulse response matrix $\mathbf{G}$ and free response $\mathbf{y}_0$ and the optimal solution of \eqref{prob:IOP_data} recovers the optimal safe controller $\mathbf{K}^\star$ for the real system. 

In practice, exact historical and recent data are not available. As per the noise model in the dynamics \eqref{eq:dynamic}-\eqref{eq:input}, one may assume that historical and recent trajectories are affected by additive noise $w^h(t),w^r(t),v^h(t),v^r(t)$\footnote{where ``$w$'' and ``$v$'' denote input and output noise, respectively, and the apices $r$ and $h$ denote recent data and historical data, respectively.} at all time instants, with zero expected values and variances $\bm{\Sigma}^h_w,\bm{\Sigma}^r_w,\bm{\Sigma}^h_v,\bm{\Sigma}^r_v$ respectively. Hence, the matrix on the left-hand-side of \eqref{eq:BM} becomes full row-rank almost surely, and any solution to \eqref{eq:BM} leads to potentially different estimates of the system free and impulse responses, which do not necessarily match the exact ones. \color{black} 
This issue is well-known in the behavioral theory literature, and several mitigation strategies have recently been proposed \cite{coulson2019data,coulson2021distributionally,de2019formulas,yin2021maximum,alpago2020extended,iannelli2020experiment}. For instance, a behavioral LS estimator akin to  the impulse-response identification of \cite{oymak2019non,zheng2020sample} is given by
\begin{equation}
\footnotesize
   \begin{bmatrix}G_{LS}&g_{LS}\end{bmatrix} \hspace{-0.1cm}=\hspace{-0.1cm} \begin{bmatrix}\hat{U}_p\\ \hat{Y}_p\\\hat{U}_f\end{bmatrix}^{+}\hspace{-0.1cm}
    \begin{bmatrix}0_{mT_{ini} \times m}& \mathbf{u}^r_{[0,T_{ini}-1]}\\
    0_{pT_{ini} \times m}& \mathbf{y}^r_{[0,T_{ini}-1]}\\
    \begin{bmatrix}I_m&0_{ m \times m(N-1)} \end{bmatrix}^\mathsf{T} & 0_{mN \times 1}
    \end{bmatrix}, \label{eq:G_LS}
\end{equation}
while the ML estimator \cite{yin2021maximum} is computed as
{\footnotesize
\begin{alignat}{3}
    G_{ML}=&\argmin_{G} &&-log\left[p\left(\begin{bmatrix}\Xi_y\\Y_fG\end{bmatrix}|~G,Y_f\right)\right]  \label{eq:G_ML}\\
    &\st&&\begin{bmatrix}\hat{U}_p\\ \hat{U}_f\end{bmatrix}
    G=\begin{bmatrix}
    0_{mT_{ini} \times m} \nonumber\\
    \begin{bmatrix}I_m&0_{ m \times m(N-1)} \end{bmatrix}^\mathsf{T}
    \end{bmatrix}\,,\\
    g_{ML}=&  \argmin_{g} && -log\left[p\left(\begin{bmatrix}\xi_y\\Y_fg\end{bmatrix}|~g,Y_f\right)\right] \label{eq:g_ML}\\
    &\st&&\begin{bmatrix}\hat{U}_p\\ \hat{U}_f\end{bmatrix}g=\begin{bmatrix}\mathbf{u}^r_{[0,T_{ini}-1]}\\0_{mN \times 1}\end{bmatrix}\,, \nonumber
\end{alignat}
}%
where the residuals $\Xi_y =(Y_p-\hat{Y}_p)G $ and $\xi_y =(Y_p-\hat{Y}_p)g$ denote the fitting deviation from the most recent output measurements, and $p(a|b)$ indicates the probability of event $a$ conditioned to $b$.

\subsection{Proof of Proposition~\ref{prop:IOP}}
\label{app:prop:IOP}
    For the first statement, notice that the controller $\mathbf{K}$ achieves the closed-loop responses \eqref{eq:CL_responses}. Now select $(\bm{\Phi}_{yy}, \bm{\Phi}_{yu}, \bm{\Phi}_{uy}, \bm{\Phi}_{uu})$ as
    \begin{equation}
    \begin{bmatrix}
           \bm{\Phi}_{yy} & \bm{\Phi}_{yu} \\
            \bm{\Phi}_{uy} & \bm{\Phi}_{uu} 
        \end{bmatrix}=\begin{bmatrix}(I-\mathbf{GK})^{-1} & (I-\mathbf{GK})^{-1}\mathbf{G}\\ \mathbf{K}(I-\mathbf{GK})^{-1} & (I-\mathbf{KG})^{-1}\end{bmatrix}\,,
        \label{eq:CL_responses_proofs}
    \end{equation}
  and $\mathbf{q} = \bm{\Phi}_{uu} \mathbf{g}$. Clearly, $\mathbf{K} = \bm{\Phi}_{uy}\bm{\Phi}_{yy}^{-1}$ and $\mathbf{g}= \bm{\Phi}_{uu}^{-1}\mathbf{q}$, and by plugging the corresponding expressions, we verify that \eqref{eq:ach} and \eqref{eq:ach3} are satisfied. It remains to prove that \eqref{eq:ach_constraints1}-\eqref{eq:ach_constraints2} are satisfied. In \eqref{eq:compact_constraints}, substitute $\mathbf{y}$ and $\mathbf{u}$ with their closed-loop expressions \eqref{eq:IOP_parameters}. It follows that the  addends separately depend on $\mathbf{w}$ or $\mathbf{v}$. Hence, \eqref{eq:compact_constraints} can be rewritten as
  {
   \begin{align}
       &\max_{\norm{\mathbf{v}}_\infty \leq v_\infty} \left(\mathbf{F}_y\bm{\Phi}_{yy}\right)\mathbf{v} + \max_{\norm{\mathbf{w}}_\infty \leq w_\infty}\left(\mathbf{F}_y\bm{\Phi}_{yu}\right)\mathbf{w}+\nonumber\\&+\mathbf{F}_y\mathbf{Gq}+\left(\mathbf{F}_y\bm{\Phi}_{yy}\right)\mathbf{CP}_A(:,0)x_0 \leq \mathbf{b}_y\,,    \label{eq:compact_constraints_separated1}
       \end{align}
       \begin{align}
       &\max_{\norm{\mathbf{v}}_\infty \leq v_\infty} \left(\mathbf{F}_u\bm{\Phi}_{uy}\right)\mathbf{v} + \max_{\norm{\mathbf{w}}_\infty \leq w_\infty} \left(\mathbf{F}_u\bm{\Phi}_{uu}\right)\mathbf{w}+\nonumber\\
       &+\mathbf{F}_u\mathbf{q}+\left(\mathbf{F}_y\bm{\Phi}_{uy}\right)\mathbf{CP}_A(:,0)x_0 \leq \mathbf{b}_u\,,\label{eq:compact_constraints_separated2}
   \end{align}
   }%
   where the $\max(\cdot)$ is to be intended row-wise. The expressions \eqref{eq:compact_constraints_separated1}-\eqref{eq:compact_constraints_separated2} are already convex in $\bm{\Phi},\mathbf{q}$. To have a more explicit expression, similar to \cite{dean2019safely} we utilize the well-known property  that the $\norm{\cdot}_1$ and the $\norm{\cdot}_\infty$ vector norms are dual of each other \cite{boyd2004convex}, that is $k\norm{x}_1 = \max_{\norm{w}_\infty \leq k} x^\mathsf{T} w$. The result follows immediately by inspecting \eqref{eq:compact_constraints_separated1}-\eqref{eq:compact_constraints_separated2} and letting $x^\mathsf{T}$ be equal to either $F_{y,j}\bm{\Phi}_{yy}$, $F_{y,j}\bm{\Phi}_{uy}$, $F_{y,j}\bm{\Phi}_{yu}$ or $F_{y,j}\bm{\Phi}_{uu}$, and letting $k$ be equal to either $v_\infty$ or $w_\infty$.
   
   For the second statement, it is easy to notice  that $\mathbf{K}$ is causal by construction because $\bm{\Phi}_{uy}$ and $\bm{\Phi}_{yy}$ are block lower-triangular. 
   By selecting the controller $\mathbf{K} = \bm{\Phi}_{uy}\bm{\Phi}_{yy}^{-1}$ one has 
   {\small
   \begin{align*}
    (I&-\mathbf{G}\bm{\Phi}_{uy}\bm{\Phi}_{yy}^{-1})^{-1} = (I-\mathbf{G}\bm{\Phi}_{uy}(I+\mathbf{G}\bm{\Phi}_{uy})^{-1})^{-1}\\
    &= ((I+\mathbf{G}\bm{\Phi}_{uy}-\mathbf{G}\bm{\Phi}_{uy})(I+\mathbf{G}\bm{\Phi}_{uy})^{-1})^{-1}\\
    &= I +\mathbf{G}\bm{\Phi}_{uy} = \bm{\Phi}_{yy}\,,
    \end{align*}
    }which shows that $\bm{\Phi}_{yy}$ is the closed-loop response from $\mathbf{v}_{[0,N-1]}+\mathbf{CP}_A(:,0)x_0$ to $\mathbf{y}_{[0,N-1]}$ as per \eqref{eq:CL_responses}. Similar computations hold for the remaining closed-loop responses. For the safety constraints, select any $\bm{\Phi}$ and $\mathbf{q}$ complying with \eqref{eq:ach_constraints1}-\eqref{eq:ach_constraints2}. It is easy to verify by direct computation that, for any $\mathbf{w}$ and $\mathbf{v}$, the same input and output trajectories defined at \eqref{eq:IOP_parameters} are obtained by letting $\mathbf{K} = \bm{\Phi}_{uy}\bm{\Phi}_{yy}^{-1}$ and $\mathbf{g} = \bm{\Phi}_{uu}^{-1}\mathbf{q}$ in \eqref{eq:state_compact}, \eqref{eq:output_compact}, \eqref{eq:control_policy}. Hence, the safety constraints are satisfied for any disturbance realization.

\subsection{Proof of Proposition~\ref{pr:robust_IOP}}
\label{app:prop:robust_IOP}
We first prove that $\mathbf{K} \in \bm{\mathcal{K}} \implies \hatbm{\Phi} \in \bm{\Pi}$, where {
\begin{equation}
\label{eq:choice_Phi}
 \hatbm{\Phi}:=\begin{bmatrix}(I-\hatbf{G}\mathbf{K})^{-1} & (I-\hatbf{G}\mathbf{K})^{-1}\hatbf{G}\\ \mathbf{K}(I-\hatbf{G}\mathbf{K})^{-1} & (I-\mathbf{K}\hatbf{G})^{-1}\end{bmatrix}\,.
\end{equation}}%
Let us fix $(\bm{\Delta},\bm{\delta}_0) \in \mathcal{E}$. By substitution of  $\hatbm{\Phi}$ inside the blocks of $\bm{\Phi}$ defined in the proposition statement, one has 
${\bm{\Phi}=\begin{bmatrix}(I-(\hatbf{G}+\bm{\Delta})\mathbf{K})^{-1} & (I-(\hatbf{G}+\bm{\Delta})\mathbf{K})^{-1}(\hatbf{G}+\bm{\Delta})\\ \mathbf{K}(I-(\hatbf{G}+\bm{\Delta})\mathbf{K})^{-1} & (I-\mathbf{K}(\hatbf{G}+\bm{\Delta}))^{-1}\end{bmatrix}}$. From \eqref{eq:IOP_parameters}-\eqref{eq:CL_responses} the closed-loop trajectories obtained by applying $\mathbf{K}$ to the system $\hatbf{G}+\bm{\Delta}$ are given by 
{
\begin{equation*}
 \begin{bmatrix}\mathbf{y}(\mathbf{K,\bm{\theta}})\\\mathbf{u}(\mathbf{K,\bm{\theta}})\end{bmatrix}=\begin{bmatrix}
           \bm{\Phi}_{yy} & \bm{\Phi}_{yu} \\
            \bm{\Phi}_{uy} & \bm{\Phi}_{uu} 
        \end{bmatrix}\begin{bmatrix}\mathbf{v}+\hatbf{y}_0+\bm{\delta}_0\\ \mathbf{w}\end{bmatrix}\,.
\end{equation*}
}%
Proceeding as in the proof of Proposition~\ref{prop:IOP}, one can show that ``$(\mathbf{y}(\mathbf{K,\bm{\theta}}),\mathbf{u}(\mathbf{K,\bm{\theta}})) \in \bm{\Gamma}$'' for every $\bm{\theta} \in \bm{\mathcal{E}} \times \bm{\mathcal{W}} \times \bm{\mathcal{V}}$ is the same as ``\eqref{eq:constraints_delta1}-\eqref{eq:constraints_delta2}'' for every $(\bm{\Delta},\bm{\delta}_0) \in \bm{\mathcal{E}}$. Since \eqref{eq:constraints_hatPhi1} and \eqref{eq:constraints_hatPhi3} are verified by construction, the proof is concluded. Further, for any $(\bm{\Delta},\bm{\delta}_0)$, the cost of \eqref{prob:robust_deltas} achieved by $\mathbf{K}$ is identical to the cost of \eqref{prob:robust_IOP} achieved by $\hatbm{\Phi}$ as proven in Proposition~\ref{prop:strongly_convex}.

Next, we show $\hatbm{\Phi} \in \bm{\Pi}\implies \hatbf{K} \in \bm{\mathcal{K}}$, where $\hatbf{K} := \hatbm{\Phi}_{uy}\hatbm{\Phi}_{yy}^{-1}$. Using \eqref{eq:constraints_hatPhi1}, one can verify that
{
\begin{equation*}
    \bm{\Phi}_{yy} := \hatbm{\Phi}_{yy}(I-\bm{\Delta}\hatbm{\Phi}_{uy})^{-1} = (I-(\hatbf{G}+\bm{\Delta})\hatbf{K})^{-1}\,,
\end{equation*}
}%
and similarly, that all other equalities in \eqref{eq:choice_Phi} hold by substituting $\hatbf{\Phi}$ with $\mathbf{\Phi}$ and $\mathbf{K}$ with $\hatbf{K}$.  Then
{
\begin{equation*}
 \begin{bmatrix}\mathbf{y}(\hatbf{K},\bm{\theta})\\\mathbf{u}(\hatbf{K},\bm{\theta})\end{bmatrix}=\begin{bmatrix}
           \bm{\Phi}_{yy} & \bm{\Phi}_{yu} \\
            \bm{\Phi}_{uy} & \bm{\Phi}_{uu} 
        \end{bmatrix}\begin{bmatrix}\mathbf{v}+\hatbf{y}_0+\bm{\delta}_0\\ \mathbf{w}\end{bmatrix}\,.
\end{equation*}
}%
But ``\eqref{eq:constraints_delta1}-\eqref{eq:constraints_delta2}'' for every $(\bm{\Delta},\bm{\delta}_0) \in \bm{\mathcal{E}}$, which hold by definition, imply that  $(\mathbf{y}(\hatbf{K},\bm{\theta}),\mathbf{u}(\hatbf{K},\bm{\theta})) \in \bm{\Gamma}$ for every $\bm{\theta} \in \bm{\mathcal{E}} \times \bm{\mathcal{W}} \times \bm{\mathcal{V}}$ (see the proof of Proposition~\ref{prop:IOP}). Further, for any $(\bm{\Delta},\bm{\delta}_0)$, the cost of \eqref{prob:robust_IOP}  achieved by $\hatbm{\Phi}$ is identical to the cost of \eqref{prob:robust_deltas} achieved by $\hatbf{K}$ as proven in Proposition~\ref{prop:strongly_convex}.

\subsection{Proof of Lemma~\ref{le:upperbound}}
\label{app:le:upperbound}
The objective function in Proposition~\ref{pr:robust_IOP} can be written as the square-root of the sum of the square of the Frobenius norms of each of its six blocks. For the upper-left block,  since $\norm{\hat{\bm{\Phi}}_{u y}}_2 \leq$ $\gamma< \epsilon_2^{-1}$ by assumption, \color{black}we have
{\small
\begin{align*}
    &\|\hatbm{\Phi}_{yy}(I-\mathbf{\Delta}\hatbm{\Phi}_{uy})^{-1}\|_F \leq \|\hatbm{\Phi}_{yy}\|_F \norm{\sum_{k=0}^\infty (\bm{\Delta} \hatbm{\Phi}_{uy})^k}_2 \\&\leq \|\hatbm{\Phi}_{yy}\|_F \sum_{k=0}^\infty \norm{\epsilon_2 \hatbm{\Phi}_{uy}}^k_2 = \|\hatbm{\Phi}_{yy}\|_F\left(1-\epsilon_2 \|\hatbm{\Phi}_{uy}\|_2\right)^{-1}\,,
\end{align*}
}
where the convergence of the series follows from $\bm{\Delta}$ and $\hatbm{\Phi}_{uy}$ having zero-entries diagonal blocks by construction. Similarly,
{\small
\begin{align*}
    &\|\hatbm{\Phi}_{uy}(I-\mathbf{\Delta}\hatbm{\Phi}_{uy})^{-1}\|_F \leq \|\hatbm{\Phi}_{uy}\|_F\left(1-\epsilon_2 \|\hatbm{\Phi}_{uy}\|_2\right)^{-1}\,,\\
    &\|(I-\hatbm{\Phi}_{uy}\mathbf{\Delta})^{-1}\hatbm{\Phi}_{uu}\|_F \leq \|\hatbm{\Phi}_{uu}\|_F\left(1-\epsilon_2 \|\hatbm{\Phi}_{uy}\|_2\right)^{-1}\,.
\end{align*}
}
Next, we have
{\small
\begin{align*}
       &\|\hatbm{\Phi}_{yy}(I-\mathbf{\Delta}\hatbm{\Phi}_{uy})^{-1}(\hatbf{G}+\mathbf{\Delta})\|_F
       \\
       &\leq \|\hatbm{\Phi}_{yy}\hatbf{G}\|_F \text{+}  \|\hatbm{\Phi}_{yy}\mathbf{\Delta}\|_F \text{+} \left\|\hatbm{\Phi}_{yy}\left(\sum_{k=1}^\infty(\mathbf{\Delta}\hatbm{\Phi}_{uy})^k\right) (\hatbf{G}+\mathbf{\Delta})\right\|_F \\
       & \leq \|\hatbm{\Phi}_{yu}\|_F + \epsilon_2 \|\hatbm{\Phi}_{yy}\|_F + \|\hatbm{\Phi}_{yy}\|_F \frac{\epsilon_2 \|\hatbm{\Phi}_{uy}\|_2(\|\hatbf{G}\|_2 + \epsilon_2)}{1 - \epsilon_2 \|\hatbm{\Phi}_{uy}\|_2}\\
       & \leq  \frac{\|\hatbm{\Phi}_{yu}\|_F + \epsilon_2 \|\hatbm{\Phi}_{yy}\|_F (2 + \|\hatbm{\Phi}_{uy}\|_2\|\hatbf{G}\|_2 )}{1 - \epsilon_2 \|\hatbm{\Phi}_{uy}\|_2}\,,
\end{align*}
}
and therefore, by developing the squares and  using that  $\norm{\hatbm{\Phi}_{yy}\hatbf{G}}_F\leq \|\hatbm{\Phi}_{yy}\|_F\|\hatbf{G}\|_2$ we obtain
{\small
\begin{align*}
        &\quad \|\hatbm{\Phi}_{yy}(I-\mathbf{\Delta}\hatbm{\Phi}_{uy})^{-1}(\hatbf{G}+\mathbf{\Delta})\|_F^2\\
       & \leq \frac{\left(\|\hatbm{\Phi}_{yu}\|_F^2
        + \|\hatbm{\Phi}_{yy}\|_F^2 h(\epsilon_2,\gamma,\hatbf{G})\right)}{(1-\epsilon_2 \|\hatbm{\Phi}_{uy}\|_2)^2}\,.
\end{align*}
}
Proceeding analogously, one can also prove that
{\small
\begin{align*}
    & \|\hatbm{\Phi}_{yy}(I-\mathbf{\Delta}\hatbm{\Phi}_{uy})^{-1}(\hatbm{y}_0+\bm{\delta}_0)\|_F^2 \\
    &\leq\frac{1}{(1-\epsilon_2 \|\hatbm{\Phi}_{uy}\|_2)^2}\left(\|\hatbm{\Phi}_{yy} \hatbf{y}_0\|_F^2 + \|\hatbm{\Phi}_{yy}\|_F^2 h(\epsilon_2,\gamma,\hatbf{y}_0)\right) \,,\\
    & \|\hatbm{\Phi}_{uy}(I-\mathbf{\Delta}\hatbm{\Phi}_{uy})^{-1}(\hatbm{y}_0+\bm{\delta}_0)\|_F^2 \leq\\
    &\frac{1}{(1-\epsilon_2 \|\hatbm{\Phi}_{uy}\|_2)^2}\left(\|\hatbm{\Phi}_{uy} \hatbf{y}_0\|_F^2 + \|\hatbm{\Phi}_{uy}\|_F^2 h(\epsilon_2,\gamma,\hatbf{y}_0)\right) \,.
\end{align*}
}Therefore, combining the above inequalities we finally conclude that
{\small
\begin{align*}
    &J(\mathbf{G},\mathbf{K})
    \leq \frac{\Big(\norm{\begin{bmatrix}
              \hatbm{\Phi}_{yy} & \hatbm{\Phi}_{yu} & \hatbm{\Phi}_{yy}\hatbf{y}_0\\\hatbm{\Phi}_{uy} & \hatbm{\Phi}_{uu} & \hatbm{\Phi}_{uy}\hatbf{y}_0
             \end{bmatrix}}_F^2 +
    \|\hatbm{\Phi}_{yy}\|_F^2 (h(\epsilon_2,\gamma,\hatbf{G})+}{1-\epsilon_2 \|\hatbm{\Phi}_{uy}\|_2}\\& \qquad \frac{+
    h(\epsilon_2,\gamma,\hatbf{y}_0)) + 
    \|\hatbm{\Phi}_{uy}\|_F^2 h(\epsilon_2,\gamma,\hatbf{y}_0)\Big)^{\frac{1}{2}}}{1-\epsilon_2 \|\hatbm{\Phi}_{uy}\|_2}\,.
\end{align*}
}

\subsection{Proof of Lemma~\ref{le:quasiconvex_safetyconstraints}}
\label{app:le:quasiconvex_safetyconstraints}
By using the fact that for $x \in \mathbb{R}^n$ and $y\in \mathbb{R}^m$ we have that $\norm{\begin{bmatrix}x^\mathsf{T}&y^\mathsf{T}\end{bmatrix}}_1 = \norm{x^{\mathsf{T}}}_1 + \norm{y^{\mathsf{T}}}_1$, the left-hand-sides of \eqref{eq:constraints_delta1}-\eqref{eq:constraints_delta2} are each made of three addends. The proof hinges on upper bounding each one of them for a generic $(\bm{\Delta,\bm{\delta}_0}) \in\bm{\mathcal{E}}$. We report the full derivations for the most informative of them.  Exploiting Holder's inequality and the relation $\left\|I-\bm{\Delta} \hatbm{\Phi}_{u y}\right\|_{\infty} \leq \frac{1}{1-\epsilon_\infty \tau}$, which can be derived by proceeding as in the proof of Lemma~\ref{le:upperbound}, we have
{
\begin{align*}
    &v_\infty\norm{F_{y,j} \hatbm{\Phi}_{yy}(I-\mathbf{\Delta}\hatbm{\Phi}_{uy})^{-1}}^{\star}_1\\ \leq
    & v_\infty \norm{F_{y,j} \hatbm{\Phi}_{yy}}^{\star}_1 \norm{(I-\bm{\Delta} \hatbm{\Phi}_{uy})^{-1}}_\infty\leq \frac{v_\infty \norm{F_{y,j} \hatbm{\Phi}_{yy}}^{\star}_1}{1-\epsilon_\infty \tau}\,,
\end{align*}
}
\color{black}
which is equal to $f_{1,j}(\hatbm{\Phi})$. Next, recalling $\hatbm{\Phi}_{yu} = \hatbm{\Phi}_{yy} \hatbf{G}$,
{\small
\begin{align*}
   & w_\infty\norm{F_{y,j} \hatbm{\Phi}_{yy}(I-\mathbf{\Delta}\hatbm{\Phi}_{uy})^{-1}(\hatbf{G}+\mathbf{\Delta})}^{\star}_1\\
   & \leq w_\infty  \norm{F_{y,j}\hatbm{\Phi}_{yu}}^{\star}_1 + \max_{\norm{\mathbf{w}}_\infty \leq w_\infty}|F_{y,j}\hatbm{\Phi}_{yy}\bm{\Delta}\mathbf{w}| +\\
   &~~~+\max_{\norm{\mathbf{w}}_\infty \leq w_\infty} \lvert
   F_{y,j}\hatbm{\Phi}_{yy}\bm{\Delta}\hatbm{\Phi}_{uy}(I-\bm{\Delta}\hatbm{\Phi}_{uy})^{-1}(\hatbf{G}+\bm{\Delta})\mathbf{w}
   \rvert\\
    &\leq w_\infty  \norm{F_{y,j}\hatbm{\Phi}_{yu}}^{\star}_1 + 
    w_\infty \epsilon_\infty \norm{F_{y,j} \hatbm{\Phi}_{yy}}^{\star}_1
    + \\
    &~~~+w_\infty \epsilon_\infty \norm{F_{y,j} \hatbm{\Phi}_{yy}}^{\star}_1 \norm{\hatbm{\Phi}_{uy}(I-\bm{\Delta}\hatbm{\Phi}_{uy})^{-1}(\hatbf{G}+\bm{\Delta})}_\infty\\
   & \leq w_\infty  \norm{F_{y,j}\hatbm{\Phi}_{yu}}^{\star}_1 \hspace{-0.2cm}\text{+} w_\infty \epsilon_\infty \norm{F_{y,j} \hatbm{\Phi}_{yy}}^{\star}_1\hspace{-0.1cm}\left(\hspace{-0.1cm}1 \text{+} \tau\frac{ \norm{\hatbf{G}}_\infty+\epsilon_\infty }{1-\epsilon_\infty \tau}\right)\hspace{-0.1cm}\\
   &=w_\infty  \norm{F_{y,j}\hatbm{\Phi}_{yu}}^{\star}_1 \hspace{-0.2cm}+ w_\infty \epsilon_\infty \norm{F_{y,j} \hatbm{\Phi}_{yy}}^{\star}_1\hspace{-0.1cm}\left(\frac{1+\tau \norm{\hatbf{G}}_\infty}{1-\epsilon_\infty \tau}\right)\\
   &=f_{2,j}(\hatbm{\Phi})\,.
\end{align*}
}
Lastly, remembering that $\hatbm{\Phi}_{uu} = I + \hatbm{\Phi}_{uy} \hatbf{G}$ and noticing that
{\small
\begin{equation*}
(I-\hatbm{\Phi}_{uy}\mathbf{\Delta})^{-1}\hatbm{\Phi}_{uu}= \hatbm{\Phi}_{uu}+\hatbm{\Phi}_{uy}\bm{\Delta}(I-\hatbm{\Phi}_{uy}\bm{\Delta})^{-1}\hatbm{\Phi}_{uu}\,,
\end{equation*}
}
we have
{\small
\begin{align*}
    &w_\infty\norm{F_{u,j} (I-\hatbm{\Phi}_{uy}\mathbf{\Delta})^{-1}\hatbm{\Phi}_{uu}}^{\star}_1 \\
    &\leq w_\infty \norm{F_{u,j} \hatbm{\Phi}_{uu}}^{\star}_1 \\
    &+w_\infty \epsilon_\infty \frac{\norm{F_{u,j} \hatbm{\Phi}_{uy}}^{\star}_1}{1-\epsilon_\infty \norm{\hatbm{\Phi}_{uy}}_\infty}
    \left(
    \norm{\hatbf{G}}_\infty \norm{\hatbm{\Phi}_{uy}}_\infty + 1\right)\\
    &\leq w_\infty\hspace{-0.1cm} \norm{F_{u,j} \hatbm{\Phi}_{uu}}^{\star}_1+w_\infty \epsilon_\infty\norm{F_{u,j} \hatbm{\Phi}_{uy}}^{\star}_1\frac{1+\tau \norm{\hatbf{G}}_\infty}{1-\epsilon_\infty \tau}\\
    &= f_{5,j}(\hatbm{\Phi})\,.
\end{align*}
}Similar computations allows one to derive the upper bounds for the remaining terms.

\subsection{Proof of Lemma~\ref{le:feasible}}
\label{app:le:feasible}

First, it is easy to verify that $\tildebm{\Phi}$ satisfies the constraints in~\eqref{prob:quasi_convex_alpha}; indeed, $\tildebm{\Phi}$ comprises the closed-loop responses when we apply $\mathbf{K}^c$ to the estimated plant $\hatbf{G}$. Next, we have 
{\small
\begin{align}
    \normtwo{\tildebm{\Phi}_{uy}} &= \normtwo{\cbm{\Phi}_{uy}(I+\bm{\Delta} \cbm{\Phi}_{uy})^{-1}} \nonumber\\
    &\leq \frac{\normtwo{\cbm{\Phi}_{uy}}}{1-\epsilon_2 \normtwo{\cbm{\Phi}_{uy}}} \leq  \sqrt{2}\frac{\normtwo{\cbm{\Phi}_{uy}}}{1-\epsilon_2 \normtwo{\cbm{\Phi}_{uy}}}\nonumber\\
    &\leq \sqrt{2}\frac{\normtwo{\starbm{\Phi}_{uy}}}{1-\epsilon_2 \normtwo{\starbm{\Phi}_{uy}}}\nonumber= \sqrt{2} \frac{\eta}{\epsilon_2(1-\eta)} = \widetilde{\gamma}\,.\nonumber
\end{align}
}%
Since $\alpha \in [\sqrt{2}\frac{\eta}{\epsilon_2(1-\eta)},\epsilon_2^{-1})$ and  $\eta<\frac{1}{5}$, then $ \tilde{\gamma} \leq \alpha < \epsilon_2^{-1}$. Hence $\widetilde{\gamma}$ is feasible. Similarly,
{
\begin{align}
   & \norm{\tildebm{\Phi}_{uy}}_\infty = \norm{\cbm{\Phi}_{uy}(I+\bm{\Delta} \cbm{\Phi}_{uy})^{-1}}_\infty \nonumber\\
    &\leq \frac{\norm{\cbm{\Phi}_{uy}}_\infty}{1-\epsilon_\infty \norm{\cbm{\Phi}_{uy}}_\infty} \leq  \frac{\norm{\starbm{\Phi}_{uy}}_\infty}{1-\epsilon_\infty \norm{\starbm{\Phi}_{uy}}_\infty} \nonumber= \frac{\zeta}{\epsilon_\infty(1-\zeta)} = \widetilde{\tau}\,.\nonumber
\end{align}
}%
Since $\zeta<\frac{1}{2}$, then $ \tilde{\tau}  < \epsilon_\infty^{-1}$ and hence it is a feasible value for $\tau$. It remains to show that $\tildebm{\Phi}$ satisfies the safety constraints \eqref{eq:quasiconvex_constraints1}-\eqref{eq:quasiconvex_constraints2}.  We know that $\cbm{\Phi}$ is feasible for \eqref{prob:doubly_robust}, and hence $\phi_{1,j}(\cbm{\Phi}) + \phi_{2,j}(\cbm{\Phi}) + \phi_{3,j}(\cbm{\Phi}) \leq \mathbf{b}_{y,j}$ and $\phi_{4,j}(\cbm{\Phi}) + \phi_{5,j}(\cbm{\Phi}) + \phi_{6,j}(\cbm{\Phi}) \leq \mathbf{b}_{u,j}$.  We conclude the proof by showing that  $f_{i,j}(\tildebm{\Phi})\leq \phi_{i,j}(\cbm{\Phi})$ for every $i = 1,\dots, 6$.  We report the full derivations for the most informative terms. 
{
\begin{align*}
    f_{1,j}(\tildebm{\Phi}) &= \frac{v_\infty \norm{F_{y,j} \left(\cbm{\Phi}_{yy}-\cbm{\Phi}_{yy}\bm{\Delta}\cbm{\Phi}_{uy}\left(I+\bm{\Delta}\cbm{\Phi}_{uy}\right)^{-1}\right)}^{\star}_1}{1-\epsilon_\infty \widetilde{\tau}}\\
    &\leq\frac{ v_\infty \norm{F_{y,j}\cbm{\Phi}_{yy}}^{\star}_1 + \frac{v_\infty\epsilon_\infty \norm{F_{y,j} \cbm{\Phi}_{yy}}^{\star}_1 \norm{\cbm{\Phi}_{uy}}_\infty}{1-\epsilon_\infty \norm{\cbm{\Phi}_{uy}}_\infty}}{1-\epsilon_\infty \widetilde{\tau}}\\
    &\leq\frac{ v_\infty \norm{F_{y,j}\cbm{\Phi}_{yy}}^{\star}_1 + \frac{v_\infty\epsilon_\infty \norm{F_{y,j} \cbm{\Phi}_{yy}}^{\star}_1 \norm{\starbm{\Phi}_{uy}}_\infty}{1-\epsilon_\infty \norm{\starbm{\Phi}_{uy}}_\infty}}{1-\epsilon_\infty \widetilde{\tau}}\\
    &\leq \frac{v_\infty \norm{F_{y,j}\cbm{\Phi}_{yy}}^{\star}_1}{1-2\zeta} = \phi_{1,j}(\cbm{\Phi})\,.
\end{align*}
}Similarly, it is easy to show that $f_{4,j}(\tildebm{\Phi})\leq \phi_{4,j}(\cbm{\Phi})$. Next, recalling \eqref{eq:suboptimal} and observing that\color{black}
{
\begin{align*}
    &\tildebm{\Phi}_{yu} = \cbm{\Phi}_{yu} -\cbm{\Phi}_{yy} \bm{\Delta} - \cbm{\Phi}_{yy} \bm{\Delta}\cbm{\Phi}_{uy}(I+\bm{\Delta}\cbm{\Phi}_{uy})^{-1}\hatbf{G}\,,\\
    & \tildebm{\Phi}_{yy} = \cbm{\Phi}_{yy}-\cbm{\Phi}_{yy}\bm{\Delta}\cbm{\Phi}_{uy}\left(I+\bm{\Delta}\cbm{\Phi}_{uy}\right)^{-1}\,,
\end{align*}
}we have
{\small
\begin{align*}
    &f_{2,j}(\tildebm{\Phi})\\
    &\leq w_\infty \norm{F_{y,j}\cbm{\Phi}_{yy}(I+\bm{\Delta}\cbm{\Phi}_{uy})^{-1}(\mathbf{G}-\bm{\Delta})}^{\star}_1+\\
    & ~~+ w_\infty\epsilon_{\infty} \norm{F_{y,j}\cbm{\Phi}_{yy}(I+\bm{\Delta}\cbm{\Phi}_{uy})^{-1}}^{\star}_1  \left(\frac{1+\widetilde{\tau}\norm{\hatbf{G}}_\infty}{1-\epsilon_\infty \widetilde{\tau}}\right)\\
    &\leq w_\infty \hspace{-0.1cm}\norm{F_{y,j}\cbm{\Phi}_{yu}}^{\star}_1\hspace{-0.15cm} \text{+} w_\infty \epsilon_\infty\hspace{-0.1cm} \norm{F_{y,j}\cbm{\Phi}_{yy}}^{\star}_1\hspace{-0.1cm}\left(\hspace{-0.1cm}1\text{+}\frac{\norm{\cbm{\Phi}_{uy}}_\infty \hspace{-0.1cm}\norm{\hatbf{G}}_\infty}{1-\epsilon_\infty \norm{\cbm{\Phi}_{uy}}_\infty}\hspace{-0.1cm}\right) \hspace{-0.1cm}\text{+}\\
    &~~+w_\infty \epsilon_\infty\frac{\norm{F_{y,j}\cbm{\Phi}_{yy}}^{\star}_1\left(\frac{1+\widetilde{\tau}\norm{\hatbf{G}}_\infty}{1-\epsilon_\infty \widetilde{\tau}}\right)}{1-\epsilon_\infty \norm{\cbm{\Phi}_{uy}}_\infty}\leq w_\infty \hspace{-0.1cm}\norm{F_{y,j}\cbm{\Phi}_{yu}}^{\star}_1 +\\
    &~~+w_\infty \norm{F_{y,j}\cbm{\Phi}_{yy}}^{\star}_1\hspace{-0.1cm}\left(\epsilon_\infty+\frac{\zeta \norm{\hatbf{G}}_\infty}{1-\zeta} + \frac{\epsilon_\infty +\frac{\zeta \norm{\hatbf{G}}_\infty}{1-\zeta}}{\left(1-\frac{\zeta}{1-\zeta}\right)\hspace{-0.1cm}(1-\zeta)}\right)\\
    &=w_\infty \hspace{-0.1cm}\norm{F_{y,j}\cbm{\Phi}_{yu}}^{\star}_1 \hspace{-0.1cm}+2w_\infty\hspace{-0.1cm} \norm{F_{y,j}\cbm{\Phi}_{yy}}^{\star}_1\hspace{-0.1cm} \frac{(1-\zeta)\left(\epsilon_\infty+\zeta\norm{\hatbf{G}}_\infty\right)}{1-2\zeta}\\
    &\leq w_\infty \hspace{-0.1cm}\norm{F_{y,j}\cbm{\Phi}_{yu}}^{\star}_1 \hspace{-0.1cm}+\hspace{-0.1cm}2w_\infty\hspace{-0.1cm} \norm{F_{y,j}\cbm{\Phi}_{yy}}^{\star}_1 \hspace{-0.1cm} \frac{\left(\epsilon_\infty+\zeta\norm{\hatbf{G}}_\infty\right)}{1-2\zeta}\\
    &\leq \phi_{2,j}(\cbm{\Phi})\,.
\end{align*}
}
\balance
Similarly, $f_{3,j}(\tildebm{\Phi})\leq \phi_{3,j}(\cbm{\Phi})$ and $f_{6,j}(\tildebm{\Phi})\leq \phi_{6,j}(\cbm{\Phi})$. By only noticing that $\norm{\cbm{\Phi}_{uu}}_\infty\leq 1+\norm{\cbm{\Phi}_{uy}}_\infty\left(\norm{\hatbf{G}}_\infty+\epsilon_\infty\right)$ and that $(1+\zeta)(1-2\zeta)\leq1-\zeta$ for every $\zeta>0$, analogous computations lead to $f_{5,j}(\tildebm{\Phi})\leq \phi_{5,j}(\cbm{\Phi})$.

\subsection{Proof of Theorem~\ref{th:suboptimality}}
\label{app:th:suboptimality}
 By denoting as $\hatbm{\Phi}^\star$ the closed-loop responses obtained by applying $\hatbf{K}^\star$ to $\hatbf{G}$, we have by Lemma~\ref{le:upperbound} and by $\gamma\leq \alpha$
{
\begin{align*}
   & J(\mathbf{G},\starhatbf{K})
             \leq \frac{1}{1-\epsilon_2 \gamma^\star} \times\\
             &\times \norm{\hspace{-0.1cm}\begin{bmatrix}
              \sqrt{1+h(\epsilon_2,\alpha,\hatbf{G}) + h(\epsilon_2,\alpha,\hatbf{y}_0)}\starhatbm{\Phi}_{yy} & \starhatbm{\Phi}_{yu} & \starhatbm{\Phi}_{yy}\hatbf{y}_0\\\sqrt{1+ h(\epsilon_2,\alpha,\hatbf{y}_0)}\starhatbm{\Phi}_{uy} & \starhatbm{\Phi}_{uu} & \starhatbm{\Phi}_{uy}\hatbf{y}_0\hspace{-0.1cm}
             \end{bmatrix}\hspace{-0.1cm}}_F\hspace{-0.15cm},
\end{align*}
}where $\gamma^\star$ is optimal for \eqref{prob:quasi_convex_alpha}.  By Lemma~\ref{le:feasible}, under the  assumptions on $\eta,\zeta,\alpha$ we have that $(\widetilde{\gamma},\widetilde{\tau},\tildebm{\Phi})$ belongs to the feasible set of \eqref{prob:quasi_convex_alpha}. Hence, by suboptimality of any feasible solution:
{
\begin{align*}
&J(\mathbf{G},\starhatbf{K})\leq \frac{1}{1-\epsilon_2 \widetilde{\gamma}}\times \\
&\times \norm{\hspace{-0.1cm}\begin{bmatrix}
              \sqrt{1+h(\epsilon_2,\alpha,\hatbf{G}) + h(\epsilon_2,\alpha,\hatbf{y}_0)}\tildebm{\Phi}_{yy} & \tildebm{\Phi}_{yu} & \tildebm{\Phi}_{yy}\hatbf{y}_0\\\sqrt{1+ h(\epsilon_2,\alpha,\hatbf{y}_0)}\tildebm{\Phi}_{uy} & \tildebm{\Phi}_{uu} & \tildebm{\Phi}_{uy}\hatbf{y}_0
             \end{bmatrix}\hspace{-0.1cm}}_F\hspace{-0.15cm}.
\end{align*}
}Using the definition of $\tildebm{\Phi}$ from Lemma~\ref{le:feasible}, we now relate
{
\begin{equation*}
    \widetilde{C}\hspace{-0.1cm} =\hspace{-0.1cm} \norm{\hspace{-0.1cm}\begin{bmatrix}
              \sqrt{1+h(\epsilon_2,\alpha,\hatbf{G}) + h(\epsilon_2,\alpha,\hatbf{y}_0)}\tildebm{\Phi}_{yy} & \tildebm{\Phi}_{yu} & \tildebm{\Phi}_{yy}\hatbf{y}_0\\\sqrt{1+ h(\epsilon_2,\alpha,\hatbf{y}_0)}\tildebm{\Phi}_{uy} & \tildebm{\Phi}_{uu} & \tildebm{\Phi}_{uy}\hatbf{y}_0
             \end{bmatrix}\hspace{-0.1cm}}_F\hspace{-0.25cm},
\end{equation*}
}%
to the optimal cost of problem~\eqref{prob:doubly_robust}. Recalling the expressions of $M^c$ and $V^c$, and similarly to Lemma~\ref{le:feasible},
 {
\begin{align*}
    &\widetilde{C} = \Big( \norm{\begin{bmatrix}
              \tildebm{\Phi}_{yy} & \tildebm{\Phi}_{yu} & \tildebm{\Phi}_{yy}\hatbf{y}_0\\\tildebm{\Phi}_{uy} & \tildebm{\Phi}_{uu} & \tildebm{\Phi}_{uy}\hatbf{y}_0
             \end{bmatrix}}_F^2 + \\
             &+\hspace{-0.1cm}\left(h(\epsilon_2,\alpha,\hatbf{G})\hspace{-0.07cm} +\hspace{-0.07cm} h(\epsilon_2,\alpha,\hatbf{y}_0)\hspace{-0.1cm}\right)\hspace{-0.1cm}\norm{\tildebm{\Phi}_{yy}}_F^2\hspace{-0.25cm}+\hspace{-0.05cm} h(\epsilon_2,\alpha,\hatbf{y}_0)\norm{\tildebm{\Phi}_{uy}}_F^2\Big)^{\frac{1}{2}}\\
             &\leq  \frac{\sqrt{J(\mathbf{G},\mathbf{K}^c)^2+ M^c \norm{\cbm{\Phi}_{yy}}_F^2 + V^c \norm{\cbm{\Phi}_{uy}}_F^2}}{1-\epsilon_2\norm{\cbm{\Phi}_{uy}}_2}\,.
\end{align*}
}

Thus, we have established the chain of inequalities
{\small
\begin{align*}
    &J(\mathbf{G},\starhatbf{K})^2 \hspace{-0.05cm}\leq \hspace{-0.05cm}\frac{\widetilde{C}^2}{(1-\epsilon_2 \widetilde{\gamma})^2} \\
    &\leq \hspace{-0.05cm} \frac{\Big(J(\mathbf{G},\mathbf{K}^c)^2+M^c \norm{\cbm{\Phi}_{yy}}_F^2 + V^c \norm{\cbm{\Phi}_{uy}}_F^2\Big)}{(1-\epsilon_2 \widetilde{\gamma})^2(1-\epsilon_2\norm{\cbm{\Phi}_{uy}}_2)^2}\,.
\end{align*}}%
Next, notice that, by definition, we have $J(\mathbf{G},\mathbf{K}^c)^2 = (S(\epsilon_\infty)+1)J(\mathbf{G},\mathbf{K}^\star)^2$.
Recalling that $\norm{\cbm{\Phi}_{uy}}_2\leq \norm{\starbm{\Phi}_{uy}}_2$ and $\norm{\cbm{\Phi}_{yy}}_2\leq \norm{\starbm{\Phi}_{yy}}_2$, observing that $\eta < \frac{1}{5}$ implies $1 - (1 + \sqrt{2})\eta \leq 2$, \color{black} and further noticing that if $M,V>0$, then $Ma^2+Vb^2 \leq (M+V)(a^2+b^2)$, we can establish:
 {\small
\begin{align*}
   & \frac{J(\mathbf{G},\starhatbf{K})^2-J(\mathbf{G},\starbf{K})^2}{J(\mathbf{G},\starbf{K})^2}
     \leq\left(\frac{1}{(1-\epsilon_2 \normtwo{\cbm{\Phi}_{uy}})^2(1-\epsilon_2 \widetilde{\gamma})^2}\right) \times\\
     &~~~\times\left(S(\epsilon_\infty)+1+\frac{M^c\norm{\cbm{\Phi}_{yy}}_F^2+V^c\norm{\cbm{\Phi}_{uy}}_F^2}{J(\mathbf{G},\mathbf{K}^\star)^2}\right)-1\\
     &\leq\hspace{-0.1cm} \left(\frac{1}{(1-\eta)^2(1-\sqrt{2}\frac{\eta}{1-\eta})^2}-1+\frac{S(\epsilon_\infty)}{(1-\eta)^2(1-\sqrt{2}\frac{\eta}{1-\eta})^2}\right)\hspace{-0.1cm}+\\
     &~~~+\frac{M^c\norm{\cbm{\Phi}_{yy}}_F^2+V^c\norm{\cbm{\Phi}_{uy}}_F^2}{(1-\eta)^2(1-\sqrt{2}\frac{\eta}{1-\eta})^2J(\mathbf{G},\mathbf{K}^\star)^2}\\
     &\leq  \eta \left(\frac{2(1+\sqrt{2})-(1+\sqrt{2})^2 \eta}{(1-(1+\sqrt{2})\eta)^2}\right) + \frac{S(\epsilon_\infty)}{(1-(1+\sqrt{2})\eta)^2}+\\
     &~~~+\frac{(M^c+V^c)J(\mathbf{G},\mathbf{K}^c)^2}{(1-(1+\sqrt{2})\eta)^2J(\mathbf{G},\mathbf{K}^\star)^2}\\
    &\leq   20\eta +4(M^c+V^c) +4S(\epsilon_\infty)(1+M^c+V^c)\,.
\end{align*}
}
Last, we prove that $20\eta + 4(M^c+V^c) = \mathcal{O}\left(\epsilon_2 \left(1+\norm{\starbm{\Phi}_{uy}}_2\right)\left(1+\norm{\mathbf{G}}_2+\norm{\mathbf{y}_0}_2\right)^2\right)$. First, notice that $M^c+V^c \leq M^\star+V^\star$, where
{\small
\begin{align*}
    M^\star &\hspace{-0.05cm}=\hspace{-0.05cm} h(\epsilon_2,\alpha,\hatbf{G})\hspace{-0.05cm} +\hspace{-0.05cm} h(\epsilon_2,\alpha,\hatbf{y}_0)\hspace{-0.05cm}\\
    &+\hspace{-0.05cm}h(\epsilon_2,\norm{\starbm{\Phi}_{uy}}_2,\mathbf{G}) \hspace{-0.05cm}+\hspace{-0.05cm} h(\epsilon_2,\norm{\starbm{\Phi}_{uy}}_2,\mathbf{y}_0)\,,\\
    V^\star &= h(\epsilon_2,\alpha,\hatbf{y}_0)+h(\epsilon_2,\norm{\starbm{\Phi}_{uy}}_2,\mathbf{y}_0)\,.
\end{align*}
}
Using $\alpha\leq 5\norm{\starbm{\Phi}_{uy}}_2$, $\eta<\frac{1}{5}$, $\norm{\hatbf{G}}_2 \leq \norm{\mathbf{G}}_2+\epsilon_2$ and $\norm{\hatbf{y}_0}_2 \leq \norm{\mathbf{y}_0}_2+\epsilon_2$, we deduce that
{\small
\begin{align*}
    &M^\star\hspace{-0.1cm} \leq 2\Big[\epsilon_2^2(2 \hspace{-0.05cm}+\hspace{-0.05cm} 5\norm{\starbm{\Phi}_{uy}}_2\|\mathbf{G}\|_2 )^2\hspace{-0.15cm}+\hspace{-0.05cm}2\epsilon_2\norm{\mathbf{G}}_2(2\hspace{-0.05cm}+\hspace{-0.05cm}5\norm{\starbm{\Phi}_{uy}}_2\norm{\mathbf{G}}_2)\\
    &+\epsilon_2^2(2 + 5\norm{\starbm{\Phi}_{uy}}_2\|\mathbf{y}_0\|_2 )^2\\
    &+2\epsilon_2\norm{\mathbf{y}_0}_2(2+5\norm{\starbm{\Phi}_{uy}}_2\norm{\mathbf{y}_0}_2)\Big] +\mathcal{O}(\epsilon_2^2)\\
    &= \mathcal{O}\left(\epsilon_2 \left(1+\norm{\starbm{\Phi}_{uy}}_2\right)\left(1+\norm{\mathbf{G}}_2+\norm{\mathbf{y}_0}_2\right)^2\right)\,,
\end{align*}
}
and, similarly, $V^\star= \mathcal{O}\left(\epsilon_2 \left(1+\norm{\starbm{\Phi}_{uy}}_2\right)\left(1+\norm{\mathbf{y}_0}_2\right)^2\right)$.

\end{document}